\documentclass[12pt, draftclsnofoot, onecolumn]{IEEEtran}
\IEEEoverridecommandlockouts

\usepackage{xcolor,colortbl}
\definecolor{col1}{HTML}{3891A6}
\definecolor{col2}{HTML}{EF5B5B}
\definecolor{col3}{HTML}{3DDC97}

\usepackage{amsmath,amsthm,relsize}
\usepackage{amsfonts, amssymb, cuted}
\usepackage{cleveref}
\usepackage{mathtools}
\usepackage{algorithm}
\usepackage[noend]{algpseudocode}
\usepackage{cite}
\usepackage{soul}
\usepackage{graphicx}
\usepackage{tikz}
\usepackage{pgfplotstable}
\usepackage{pgfplots}
\usepackage{nicefrac}
\usepackage{enumitem}
\usepackage{support-caption}
\usepackage{subcaption}
\usepackage[font=footnotesize]{caption}
\usepackage{multirow}
\pgfplotsset{compat=1.15}
\usepackage{relsize}
\usetikzlibrary{patterns}
\usepackage{acro}
\usepackage{pifont}

\newtheorem{remark}{Remark}
\usepackage{todonotes}

\newtheorem{lemma}{Lemma}

\newcommand{\herm}{^{\mbox{\scriptsize H}}}

\newcommand{\vbar}{\raisebox{.17ex}{\rule{.04em}{1.35ex}}}
\newcommand{\vbarind}{\raisebox{.01ex}{\rule{.04em}{1.1ex}}}
\newcommand{\R}{\ifmmode{\rm I}\hspace{-.2em}{\rm R} \else ${\rm I}\hspace{-.2em}{\rm R}$ \fi}
\newcommand{\T}{\ifmmode{\rm I}\hspace{-.2em}{\rm T} \else ${\rm I}\hspace{-.2em}{\rm T}$ \fi}
\newcommand{\N}{\ifmmode{\rm I}\hspace{-.2em}{\rm N} \else \mbox{${\rm I}\hspace{-.2em}{\rm N}$} \fi}
\newcommand{\B}{\ifmmode{\rm I}\hspace{-.2em}{\rm B} \else \mbox{${\rm I}\hspace{-.2em}{\rm B}$} \fi}
\newcommand{\Hil}{\ifmmode{\rm I}\hspace{-.2em}{\rm H} \else \mbox{${\rm I}\hspace{-.2em}{\rm H}$} \fi}
\newcommand{\C}{\ifmmode\hspace{.2em}\vbar\hspace{-.31em}{\rm C} \else \mbox{$\hspace{.2em}\vbar\hspace{-.31em}{\rm C}$} \fi}
\newcommand{\Cind}{\ifmmode\hspace{.2em}\vbarind\hspace{-.25em}{\rm C} \else \mbox{$\hspace{.2em}\vbarind\hspace{-.25em}{\rm C}$} \fi}
\newcommand{\Q}{\ifmmode\hspace{.2em}\vbar\hspace{-.31em}{\rm Q} \else \mbox{$\hspace{.2em}\vbar\hspace{-.31em}{\rm Q}$} \fi}
\newcommand{\Z}{\ifmmode{\rm Z}\hspace{-.28em}{\rm Z} \else ${\rm Z}\hspace{-.28em}{\rm Z}$ \fi}

\DeclareAcronym{ADMM}{
    short = ADMM,
    long = alternating direction method of multipliers,
    list = Alternating Direction Method of Multipliers,
    tag = abbrev
}

\DeclareAcronym{AoA}{
    short = AoA,
    long = angle-of-arrival,
    list = Angle-of-Arrival,
    tag = abbrev
}

\DeclareAcronym{SISO}{
    short = SISO,
    long = single-input-single-output,
    list = single-input-single-output,
    tag = abbrev
}

\DeclareAcronym{TDMA}{
    short = TDMA,
    long = time division multiple access,
    list = time division multiple access,
    tag = abbrev
}

\DeclareAcronym{MRT}{
    short = MRT,
    long = maximum ratio transmitter,
    list = maximum ratio transmitter,
    tag = abbrev
}

\DeclareAcronym{PDA}{
    short = PDA,
    long = placement delivery array,
    list = placement delivery array,
    tag = abbrev
}

\DeclareAcronym{EE}{
    short = EE,
    long = energy efficiency,
    list = energy efficiency,
    tag = abbrev
}

\DeclareAcronym{MDS}{
    short = MDS,
    long = maximum distance separation,
    list = maximum distance separation,
    tag = abbrev
}

\DeclareAcronym{SIC}{
    short = SIC,
    long = successive-interference-cancellation,
    list = successive-interference-cancellation,
    tag = abbrev
}

\DeclareAcronym{MAC}{
    short = MAC,
    long = multiple-access-channel,
    list = multiple-access-channel,
    tag = abbrev
}

\DeclareAcronym{AoD}{
    short = AoD,
    long = angle-of-departure,
    list = Angle-of-Departure,
    tag = abbrev
}

\DeclareAcronym{BB}{
    short = BB,
    long = base band,
    list = Base Band,
    tag = abbrev
}

\DeclareAcronym{BC}{
    short = BC,
    long = broadcast channel,
    list = Broadcast Channel,
    tag = abbrev
}

\DeclareAcronym{BS}{
    short = BS,
    long = base station,
    list = Base Station,
    tag = abbrev
}

\DeclareAcronym{BR}{
    short = BR,
    long = best response,
    list = Best Response, 
    tag = abbrev
}

\DeclareAcronym{CB}{
    short = CB,
    long = coordinated beamforming,
    list = Coordinated Beamforming,
    tag = abbrev
}

\DeclareAcronym{CC}{
    short = CC,
    long = coded caching,
    list = Coded Caching,
    tag = abbrev
}

\DeclareAcronym{CE}{
    short = CE,
    long = channel estimation,
    list = Channel Estimation,
    tag = abbrev
}

\DeclareAcronym{CoMP}{
    short = CoMP,
    long = coordinated multi-point transmission,
    list = Coordinated Multi-Point Transmission,
    tag = abbrev
}

\DeclareAcronym{CRAN}{
    short = C-RAN,
    long = cloud radio access network,
    list = Cloud Radio Access Network,
    tag = abbrev
}

\DeclareAcronym{CSE}{
    short = CSE,
    long = channel specific estimation,
    list = Channel Specific Estimation,
    tag = abbrev
}

\DeclareAcronym{CSI}{
    short = CSI,
    long = channel state information,
    list = Channel State Information,
    tag = abbrev
}

\DeclareAcronym{CSIT}{
    short = CSIT,
    long = channel state information at the transmitter,
    list = Channel State Information at the Transmitter,
    tag = abbrev
}

\DeclareAcronym{CU}{
    short = CU,
    long = central unit,
    list = Central Unit,
    tag = abbrev
}

\DeclareAcronym{D2D}{
    short = D2D,
    long = device-to-device,
    list = Device-to-Device,
    tag = abbrev
}

\DeclareAcronym{DE-ADMM}{
    short = DE-ADMM,
    long = direct estimation with alternating direction method of multipliers,
    list = Direct Estimation with Alternating Direction Method of Multipliers,
    tag = abbrev
}

\DeclareAcronym{DE-BR}{
    short = DE-BR,
    long = direct estimation with best response,
    list = Direct Estimation with Best Response,
    tag = abbrev
}

\DeclareAcronym{DE-SG}{
    short = DE-SG,
    long = direct estimation with stochastic gradient,
    list = Direct Estimation with Stochastic Gradient,
    tag = abbrev
}

\DeclareAcronym{DFT}{
	short = DFT,
	long = discrete fourier transform,
	list = Discrete Fourier Transform,
	tag = abbrev
}

\DeclareAcronym{DoF}{
    short = DoF,
    long = degrees of freedom,
    list = Degrees of Freedom,
    tag = abbrev
}

\DeclareAcronym{DL}{
    short = DL,
    long = downlink,
    list = Downlink,
    tag = abbrev
}

\DeclareAcronym{GD}{
	short = GD, 
	long = gradient descent,
	list = Gradeitn Descent,
	tag = abbrev
}

\DeclareAcronym{IBC}{
    short = IBC,
    long = interfering broadcast channel,
    list = Interfering Broadcast Channel,
    tag = abbrev
}

\DeclareAcronym{i.i.d.}{
    short = i.i.d.,
    long = independent and identically distributed,
    list = Independent and Identically Distributed,
    tag = abbrev
}

\DeclareAcronym{JP}{
    short = JP,
    long = joint processing,
    list = Joint Processing,
    tag = abbrev
}

\DeclareAcronym{KKT}{
    short = KKT,
    long = Karush-Kuhn-Tucker,
    tag = abbrev
}

\DeclareAcronym{LOS}{
	short = LOS,
	long = line-of-sight,
	list = Line-of-Sight,
	tag = abbrev
}

\DeclareAcronym{LS}{
    short = LS,
    long = least squares,
    list = Least Squares,
    tag = abbrev
}

\DeclareAcronym{LTE}{
    short = LTE,
    long = Long Term Evolution,
    tag = abbrev
}

\DeclareAcronym{LTE-A}{
    short = LTE-A,
    long = Long Term Evolution Advanced,
    tag = abbrev
}

\DeclareAcronym{MIMO}{
    short = MIMO,
    long = multiple-input multiple-output,
    list = Multiple-Input Multiple-Output,
    tag = abbrev
}

\DeclareAcronym{MISO}{
    short = MISO,
    long = multiple-input single-output,
    list = Multiple-Input Single-Output,
    tag = abbrev
}

\DeclareAcronym{MSE}{
    short = MSE,
    long = mean-squared error,
    list = Mean-Squared Error,
    tag = abbrev
}

\DeclareAcronym{MMSE}{
    short = MMSE,
    long = minimum mean-squared error,
    list = Minimum Mean-Squared Error,
    tag = abbrev
}

\DeclareAcronym{mmWave}{
	short = mmWave,
	long = millimeter wave,
	list = Millimeter Wave,
	tag = abbrev
}

\DeclareAcronym{MU-MIMO}{
    short = MU-MIMO,
    long = multi-user \ac{MIMO},
    list = Multi-User \ac{MIMO},
    tag = abbrev
}

\DeclareAcronym{OTA}{
    short = OTA,
    long = over-the-air,
    list = Over-the-Air,
    tag = abbrev
}

\DeclareAcronym{PSD}{
    short = PSD,
    long = positive semidefinite,
    list = Positive Semidefinite,
    tag = abbrev
}

\DeclareAcronym{QoS}{
	short = QoS,
	long = quality of service,
	list = Quality of Service,
	tag = abbrev
}

\DeclareAcronym{QoE}{
	short = QoE,
	long = quality of experience,
	list = Quality of experience,
	tag = abbrev
}

\DeclareAcronym{RCP}{
	short = RCP,
	long = remote central processor,
	list = Remote Central Processor,
	tag = abbrev
}

\DeclareAcronym{RRH}{
    short = RRH,
    long = remote radio head,
    list = Remote Radio Head,
    tag = abbrev
}

\DeclareAcronym{RSSI}{
    short = RSSI,
    long = received signal strength indicator,
    list = Received Signal Strength Indicator,
    tag = abbrev
}

\DeclareAcronym{RX}{
	short = RX,
	long = receiver,
	list = Receiver,
	tag = abbrev
}

\DeclareAcronym{SCA}{
    short = SCA,
    long = successive convex approximation,
    list = Successive Convex Approximation,
    tag = abbrev
}

\DeclareAcronym{SG}{
    short = SG,
    long = stochastic gradient,
    list = Stochastic Gradient,
    tag = abbrev
}

\DeclareAcronym{SNR}{
    short = SNR,
    long = signal-to-noise ratio,
    list = Signal-to-Noise Ratio,
    tag = abbrev
}

\DeclareAcronym{SINR}{
    short = SINR,
    long = signal-to-interference-plus-noise ratio,
    list = Signal-to-Interference-plus-Noise Ratio,
    tag = abbrev
}

\DeclareAcronym{SOCP}{
	short = SOCP, 
	long = second order cone program,
	list = Second Order Cone Program,
	tag = abbrev
}

\DeclareAcronym{SSE}{
    short = SSE,
    long = stream specific estimation,
    list = Stream Specific Estimation,
    tag = abbrev
}

\DeclareAcronym{SVD}{
	short = SVD,
	long = singular value decomposition,
	list = Singular Value Decomposition,
	tag = abbrev
}

\DeclareAcronym{TDD}{
	short = TDD,
	long = time division duplex,
	list = Time Division Duplex,
	tag = abbrev
}

\DeclareAcronym{TX}{
	short = TX,
	long = transmitter,
	list = Transmitter,
	tag = abbrev
}

\DeclareAcronym{UE}{
    short = UE,
    long = user equipment,
    list = User Equipment,
    tag = abbrev
}

\DeclareAcronym{UL}{
    short = UL,
    long = uplink,
    list = Uplink,
    tag = abbrev
}

\DeclareAcronym{ULA}{
	short = ULA,
	long = uniform linear array,
	list = Uniform Linear Array,
	tag = abbrev
}

\DeclareAcronym{UPA}{
    short = UPA,
    long = uniform planar array,
    list = Uniform Planar Array,
    tag = abbrev
}

\DeclareAcronym{WMMSE}{
    short = WMMSE,
    long = weighted minimum mean-squared error,
    list = Weighted Minimum Mean-Squared Error,
    tag = abbrev
}

\DeclareAcronym{WMSEMin}{
    short = WMSEMin,
    long = weighted sum \ac{MSE} minimization,
    list = Weighted sum \ac{MSE} Minimization,
    tag = abbrev
}

\DeclareAcronym{WBAN}{
	short = WBAN,
	long = wireless body area network,
	list = Wireless Body Area Network,
	tag = abbrev
}

\DeclareAcronym{WSRMax}{
    short = WSRMax,
    long = weighted sum rate maximization,
    list = Weighted Sum Rate Maximization,
    tag = abbrev
}

\DeclareAcronym{XR}{
    short = XR,
    long = extended reality,
    list = Extended Reality,
    tag = abbrev
}

\newtheorem{exmp}{Example}
\theoremstyle{definition}

\newcommand{\CA}[0]{{\mathcal{A}}}
\newcommand{\CB}[0]{{\mathcal{B}}}

\newcommand{\CD}[0]{{\mathcal{D}}}

\newcommand{\CI}[0]{{\mathcal{I}}}

\newcommand{\CK}[0]{{\mathcal{K}}}

\newcommand{\CQ}[0]{{\mathcal{Q}}}

\newcommand{\CS}[0]{{\mathcal{S}}}
\newcommand{\CT}[0]{{\mathcal{T}}}
\newcommand{\CU}[0]{{\mathcal{U}}}
\newcommand{\CV}[0]{{\mathcal{V}}}

\newcommand{\Bh}[0]{{\mathbf{h}}}

\newcommand{\Bv}[0]{{\mathbf{v}}}

\newcommand{\Bx}[0]{{\mathbf{x}}}

\newcommand{\subparagraph}{}
\usepackage{titlesec}
\titlespacing\section{3pt}{6pt plus 4pt minus 2pt}{6pt plus 2pt minus 2pt}
\titlespacing\subsection{3pt}{4pt plus 4pt minus 2pt}{4pt plus 2pt minus 2pt}
\titlespacing\subsubsection{3pt}{3pt plus 4pt minus 2pt}{0pt plus 2pt minus 3pt}

\title{ Multi-antenna Coded Caching for  Location-Dependent Content Delivery}
\begin{document}

\author{\IEEEauthorblockN{Hamidreza Bakhshzad Mahmoodi, MohammadJavad Salehi, and Antti T\"olli} \\
\IEEEauthorblockA{
    Centre for Wireless Communications, University of Oulu, 90570 Oulu, Finland \\
    \textrm{E-mail: \{firstname.lastname\}@oulu.fi}
    }
\thanks{Part of this work has been presented at the IEEE International Symposium on Information Theory (ISIT) 2021, in Melbourne, Australia, {and IEEE GLOBECOM 2022, Rio de Janeiro-RJ, Brazil}. This work is supported by the Academy of Finland under grants no. 318927 (6Genesis Flagship), 319059 (Coded Collaborative Caching for Wireless Energy Efficiency - CCCWEE), and 343586 (Cache-aided mmWave Access for Immersive Digital Environments - CAMAIDE), and by the Finnish Research Impact Foundation under the project Directional Data Delivery for Wireless Immersive Digital Environments (3D-WIDE).}
}

\maketitle

\begin{abstract}
Human-computer interaction continuously evolves towards a genuinely immersive experience, submerging users in a three-dimensional (3D) virtual world. A realistic, immersive experience necessitates a highly reliable and agile wireless connection to support immense data transmission. Yet, there are abundant but underutilized memory resources available at the devices which can be harnessed as supplementary assets to reduce the excessive burden on the wireless medium. What is more, the use of Coded Caching (CC) techniques enables the cumulative cache memory of users in the network to be used as an additional communication resource. To this end, a location-dependent multi-antenna CC-based content delivery scheme tailored specifically for wireless extended reality applications is proposed in this paper. First, a novel memory allocation process is developed, enabling an appropriate trade-off between local and global caching gains. In this regard, the local caching gain is maximized when the memory is mostly dedicated to locations with poor connectivity conditions (absolute fairness). In contrast, the global caching gain is maximized when the memory is uniformly allocated among all the locations. As a result of the memory allocation process, unequal fractions of location-dependent multimedia content are cached by each user. {Given the asymmetric cache placement, a novel algorithm is proposed to create suitable codewords for each user during the subsequent delivery phase, which simultaneously achieves a global and local caching gain.} The proposed delivery scheme also combines global caching and spatial multiplexing gains using a weighted max-min multicast beamformer design with multi-rate modulation. Numerical experiments and mathematical analysis demonstrate significant performance gains, in terms of the 95-percentile expected delivery time, compared to unicast and multicast scenarios where either the local or global caching gain is maximized.

\end{abstract}

\begin{IEEEkeywords}
Multi-antenna communications, coded caching, location-dependent caching, immersive viewing, extended reality, weighted-max-min beamforming. 
\end{IEEEkeywords}

\section{Introduction}
It is expected that 5G penetration will surpass the ten percent mark by 2023 while the average per-user throughput will encounter more than a ten-fold increase compared with what was achievable five years earlier with 4G-LTE~\cite{cisco2020}. This is primarily due to new data-intensive services such as wireless extended reality (XR) applications offered by 5G and beyond{~\cite{Nokia-immersive,6G_white_paper_2020,bastug2017toward, taleb2022_towards_XR_vision, walid_sad_bennis_VR_XR_2022, flashback_2016_VR_static_dynamic_support,wireless_virtual_reality_TCOM2018}.} {Wireless XR applications necessitate stringent quality of service (QoS) in terms of both low latency ($ < 10$ ms) and high throughput ($6.37 - 95.55$ Gbps)~\cite{Nokia-immersive,6G_white_paper_2020,bastug2017toward, taleb2022_towards_XR_vision, walid_sad_bennis_VR_XR_2022, flashback_2016_VR_static_dynamic_support,wireless_virtual_reality_TCOM2018}}. Indeed, supporting the high-data-rate wireless connectivity {with low latency necessitated for} such data-intensive applications {calls for} more advanced solutions than merely increasing the available bandwidth~\cite{bastug2017toward}. Meanwhile, improving caching and computing capabilities at end-users has been deemed highly effective in increasing the transmission efficiency{~\cite{flashback_2016_VR_static_dynamic_support,wireless_virtual_reality_TCOM2018,wireless_caching_analysis_TWC2015}}. As such, upcoming mobile broadband applications rely heavily on asynchronous content reuse~\cite{proactive_caching_Cm2016}, and hence, proactive caching of popular content at the end-users could relieve network congestion and bandwidth consumption during peak traffic demand times~\cite{role_of_caching_in_future_wireless_caire_JSAC2018}. In this regard, various works have considered proactive caching in a \ac{SISO} setting to demonstrate its potential~\cite{sun2019communications,yang2018communication,sun2020bandwidth,dang2019joint}. Specifically, by utilizing caching and computing capabilities of XR mobile gadgets, the traffic burden over the wireless network can be effectively alleviated. Moreover, significant bandwidth and delay-reduction gains have also been demonstrated in~\cite{sun2019communications,yang2018communication,sun2020bandwidth,dang2019joint}.

The \ac{CC} technique, initially proposed by Maddah-Ali and Niesen in~\cite{MaddahAli-2014}, has recently gained attention due to an additional \emph{global caching gain} compared to traditional (local) caching schemes. This gain is achieved by intelligent utilization of the aggregate cache memory available throughout the network. Remarkably, the global caching gain scales linearly with the total number of users in the network, making it appealing for multi-user collaborative use cases such as XR applications~\cite{salehi2022enhancing}. In this regard, a recent work in~\cite{CC_edge_computing_for_VR_twc2021} has reduced the transmission bandwidth for delay-constrained XR applications by leveraging coded cache placement and mobile devices' computing capabilities in a \ac{SISO} setup. However, an exciting property of CC schemes is their capability to combine global caching and spatial multiplexing gains resulting from multi-antenna transmissions~\cite{pooya-cc-physical-2019-journal}. This makes CC even more appealing as multi-antenna connectivity plays a critical role in future communication systems~\cite{6G_white_paper_2020}. Nevertheless, there is a gap in the literature when it comes to applying multi-antenna CC techniques to XR setups, especially taking advantage of their location-dependent content access characteristics.

In this paper, we introduce a new multi-antenna \ac{CC} delivery scheme with location-dependent content requests well-tailored for future collaborative XR applications. {In the proposed setting, a single transmitter equipped with multiple antennas has access to a library and serves a group of cache-enabled users.} We consider a wireless connectivity scenario where the users are free to move, and their requested contents depend on their instantaneous locations in the application environment. Such a scenario entails a substantial volume of multimedia traffic with guaranteed QoS throughout the operating environment. In this regard, a location-dependent, uneven memory allocation is carried out based on the approximated or predicted data rate at each given location. Specifically, the portion of memory dedicated to each location is affected by the quality of wireless connectivity at that location. This is in contrast to conventional CC schemes, where the same portion of the memory is dedicated to each file in the library, necessitating new delivery schemes to be devised. Thus, a novel packet generation scheme is introduced to handle the irregularity by creating packets with sizes proportional to the corresponding uneven cache ratios. Finally, a multicast beamforming scheme with an underlying multi-rate modulation is proposed to leverage global caching and multiplexing gains simultaneously and hence, to improve the QoS compared to the state-of-the-art.
\subsection{Prior Art}
\noindent\textbf{Single- and multi-antenna coded caching.} 
Encouraged by the appealing \ac{CC} gains, the original error-free single-server system model in~\cite{MaddahAli-2014} was later extended to various other practical scenarios such as multi-server and wireless multi-antenna coded caching~\cite{Shariatpanahi2016,pooya-cc-physical-2019-journal,tolli2017multi}. Interestingly, the global caching gain was shown to be additive with the spatial multiplexing gain when CC is applied to a multi-antenna setup~\cite{pooya-cc-physical-2019-journal}. Moreover, the optimized multi-antenna precoder design was shown to be crucial for the \ac{CC}, especially in the low \ac{SNR} regime, to account for the inter-stream interference~\cite{tolli2017multi} appropriately. Device-to-device (D2D) extensions of multi-antenna coded caching can also be found in~\cite{DPDA2019,Ji2016,ISWCSmyself,Mahmoodi-etal-Arxiv19}. Specifically, while~\cite{DPDA2019} and~\cite{Ji2016} considered an infrastructure-less network where the only available link is D2D, works~\cite{ISWCSmyself} and~\cite{Mahmoodi-etal-Arxiv19} extended this system model to a general framework where the downlink transmission is assisted with D2D links.
Meanwhile, various practical limitations of \ac{CC} were also addressed by the research community. Most notably, it is well-known that to achieve the original caching gain proposed in~\cite{MaddahAli-2014}, the underlying scheme requires splitting finite-length files into an exponentially growing number of subpackets (with respect to the network size)~\cite{lampiris2018adding}. This exponential growth is even more severe in multi-antenna setups~\cite{Shariatpanahi2016,pooya-cc-physical-2019-journal,tolli2017multi,DPDA2019,Ji2016,ISWCSmyself,Mahmoodi-etal-Arxiv19}, motivating the research on reduced-subpacketization \ac{CC} schemes with no or moderate performance loss~\cite{lampiris2018adding,salehi2020lowcomplexity,Caire_MLPDA_2023}. In a similar work, the effect of the subpacketization on the low-SNR rate was also investigated in~\cite{salehi2019subpacketization}. {Unlike~\cite{Shariatpanahi2016,pooya-cc-physical-2019-journal,tolli2017multi,DPDA2019,Ji2016,ISWCSmyself,Mahmoodi-etal-Arxiv19,lampiris2018adding,salehi2020lowcomplexity,Caire_MLPDA_2023,salehi2019subpacketization}, which consider perfect channel 
 state information at the transmitter (CSIT), authors in~\cite{Ngo2017_MISO_CC_WOCSIT} devise a scheme for imperfect CSIT that scales with the number of users. Finally, as in the CC network, a common message is being transmitted to several users, users’ privacy requirements to prevent information leakage are also addressed in the literature (e.g.,~\cite{letafati2021learning} and~\cite{sojdeh2022secure}).}

\noindent\textbf{Coded caching with multi-rate transmission.}
A less-studied problem of \ac{CC} schemes, affecting content delivery applications in general and XR applications in particular, is the \textit{near-far} issue. Specifically, due to the underlying multicasting nature of conventional \ac{CC} schemes~\cite{Shariatpanahi2016,pooya-cc-physical-2019-journal,tolli2017multi}, the achievable rate in any multicast message is limited to the rate of the user with the worst channel conditions. In fact, as studied in~\cite{zhao_petros_MU_MISO_near_far_issue_ITW2021}, the effective gains of conventional \ac{SISO}-\ac{CC} schemes could entirely vanish at the low-\ac{SNR} region due to the near-far issue. 
To address this shortcoming, a congestion control technique was proposed in~\cite{destounis2020adaptive} to avoid serving users in adverse channel conditions, and multiple descriptor codes~(MDC) were utilized in~\cite{salehi2020coded} to serve ill-conditioned users with a lower quality of experience (QoE). Similarly, a stochastic \ac{CC} model considering queue minimization and packet control was introduced in~\cite{Coded_caching_for_stochastic_wireless_network_TWC2021}, and joint power minimization and scheduling over a wireless \ac{CC} network for delay-constrained applications was also proposed in~\cite{liu_joint_power_energi_cc_TWC2021}. Using a different perspective, it was discussed in~\cite{ozfatura_mobility_awaire_RCOM2020} that as long as user mobility patterns were known at the server, different cache profiles could be assigned to multiple cache-enabled helper nodes scattered throughout the environment to improve the \ac{CC} transmission rate. Moreover, guiding users towards locations with preferable conditions in an immersive XR application using learning-based techniques was considered in~\cite{reinforcement_for_immersive_elsevier2017}, and an order-optimal location-based coded cache placement was proposed in~\cite{caire_optimal_location_dependent_arxiv2021} to assign different cache profiles to cache-enabled transmitters located in distinct locations.

Unlike~\cite{salehi2020coded,Coded_caching_for_stochastic_wireless_network_TWC2021,liu_joint_power_energi_cc_TWC2021, ozfatura_mobility_awaire_RCOM2020, reinforcement_for_immersive_elsevier2017,caire_optimal_location_dependent_arxiv2021}, which were based on standard XOR-ing of data elements, nested code modulation (NCM) was utilized in~\cite{tang2017coded} to allow building codewords that serve every user in the multicasting group with a different rate. Several other multi-rate modulation schemes can also be found in~\cite{kramer_broadcast_channel_side_information_ITW2007, chen2010novel, tang2011full, asadi_side_information_broadcast_multi_rate_TIT2015}. The multi-rate property in these schemes was achieved by altering the modulation constellation using side information available to each user. Later on, authors in~\cite{zhao_wireless_CC_ring_area_ISIT2021} and~\cite{zhao_petros_near_far_Nakagami_asilomar2021} benefited from the shared-cache idea of~\cite{lampiris2018adding} along with the NCM scheme to compensate for the near-far problem caused by the users with adverse channel conditions. 
However, all these works are either limited to single-antenna transceivers or fixed-connectivity network conditions where the users' rates are fixed and known to the transmitter. Thus, the near-far issue still needs to be addressed in both multi-antenna setups and real-time applications where users frequently move within the network, and their achievable rate changes accordingly.


\subsection{Our contribution}
\label{section:our_contribution}
{This paper proposes a novel CC-based multi-antenna content delivery scheme for location-dependent data requests, particularly focusing on collaborative XR applications that require high data-rate connectivity and are bound to strict delay constraints. In realistic multiuser environments, available radio resources should be shared among all the users of the given XR application, limiting the available link qualities due to a higher network load~\cite{salehi2022enhancing}. In such a scenario, efficient utilization of in-device memories available to the users will be highly beneficial. In a collaborative XR setup, all users are served simultaneously in a bounded environment, where the actions and choices of each user affect the final results perceived by all users. Specifically, we follow the XR connectivity framework in~\cite{salehi2022enhancing}, where the XR content is decomposed into the so-called \textit{static}, and \textit{dynamic} parts (see Figure~\ref{fig:static_dynamic_decomposition}) and multi-antenna \ac{CC} techniques are used to deliver the {cacheable} part efficiently. In addition, cache-enabled users are scattered in the application environment and move freely. As users change their position in the environment, their achievable rate varies based on their location.} We assume {that}
the XR application environment is split into several single transmission units (STU) such that a separate 3D omnidirectional image is needed to reconstruct the virtual environment in each STU. {Each user requests the server to receive the missing data to reconstruct its field of view (FoV). After collecting all users' requests, the server transmits the missing data for dynamic and static file parts to all users (over the air) while also instructing them on reconstructing their FoVs using cache content and the delivered data.} 
Depending on the distance from the transmitter and possible infrastructure elements obstructing the wireless link in the XR environment, communication quality could vary for different STUs. Therefore, this paper aims to design caching and delivery schemes to minimize the transmission time and avoid excessive delays in serving all the XR users in a non-uniform wireless connectivity scenario with location-dependent content requests.

\begin{figure}[ht]
     \centering
      \includegraphics[height=3.7cm]{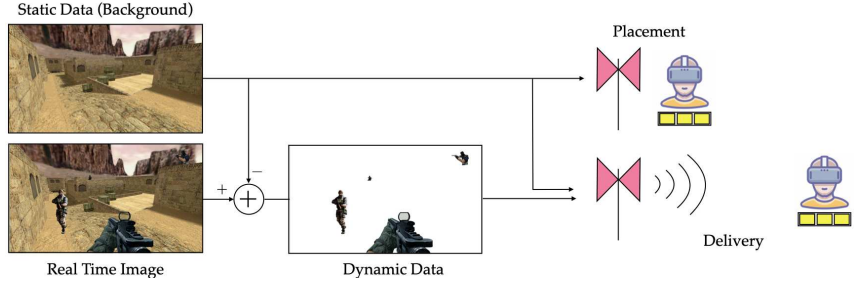}
        \caption{XR Data decomposition into static and dynamic parts.}
        \label{fig:static_dynamic_decomposition}
\end{figure}

{Intuitively, larger cache portions are allocated to the contents requested in the wireless connectivity bottlenecks to avoid excessive delivery duration and minimize the average transmission time.} In this regard, we first design a new memory allocation process that uses the available predictions of the achievable rate within the application environment to prioritize the caching of the content requested in STUs with reduced connectivity. The proposed allocation process enables a trade-off between local and global caching gains, such that the local caching gain is maximized when the memory is mostly dedicated to locations with poor connectivity conditions (absolute fairness), and the global caching gain is maximized when the memory is uniformly allocated among all the locations.
Then, for the resulting non-uniform cache allocation setup, a novel content delivery algorithm is introduced for achieving a global caching gain additive to the spatial multiplexing gain. The proposed delivery scheme relies on underlying multi-rate transmission techniques to simultaneously serve users with diverse channel conditions, i.e., to transmit smaller amounts of data to users in poor-connectivity STUs while simultaneously delivering larger amounts of data to other users. Indeed, the non-uniformity in the cache allocation causes a \ac{DoF} loss compared with the existing multi-antenna coded caching schemes that use symmetric content placement. Nevertheless, the proposed scheme is better tailored to the considered XR application scenario as it {avoids} excessive delivery {time for serving} users in areas with poor communication quality.

{The current paper is an extension of our earlier conference publications ~\cite{Mahmoodi_immersive_isit2021} and~\cite{mahmoodi2022asymmetric}. In~\cite{Mahmoodi_immersive_isit2021}, a novel location-dependent CC scheme is proposed for a single antenna transmitter, and~\cite{mahmoodi2022asymmetric} is the multi-antenna extension of~\cite{Mahmoodi_immersive_isit2021}, where we assume the memory allocation process is done such that the global caching gain at each location is an integer.
In this paper, 1) the beamforming design is described in more detail, 2) the requirement for integer global caching gains is relaxed, and 3) a two-phase delivery scheme comprising both multicast and unicast transmissions is introduced to improve the overall performance. In~\cite{mahmoodi_ICC2022_nonsym}, a similar location-dependent CC scheme was proposed based on a signal-level cache-aided interference cancellation scheme from~\cite{salehi2020lowcomplexity}. The scheme in~\cite{mahmoodi_ICC2022_nonsym} benefits from a lower subpacketization and simpler beamformer design than the scheme proposed herein. Yet, the scheme in~\cite{mahmoodi_ICC2022_nonsym} is limited to scenarios where the spatial multiplexing gain exceeds the coded caching gain (thus limiting the actual benefit of applying coded caching techniques). In addition, the finite-SNR performance of~\cite{mahmoodi_ICC2022_nonsym} is strictly inferior to the scheme proposed herein as it lacks the multicasting gain of XORing approach~\cite{MaddahAli-2014} (c.f.,~\cite{salehi2020lowcomplexity}).}  
\subsection{Notation and structure}
Matrices and vectors are presented by boldface upper and lower case letters, respectively, and calligraphic letters are used to denote sets.
For the set $\CA$ and vector $\Bv$, $|\CA|$ and $\|\Bv\|$ represent the cardinality of $\CA$ and norm of $\Bv$, respectively. Also, for two sets $\CA$ and $\CB$, $\CA \backslash \CB$ includes the elements of $\CA$ that are not in $\CB$.
{Moreover}, $[m]$ denotes the set of integer numbers $\{1, . . . ,m\}$, and $\oplus$ represents addition in the corresponding finite field. {Finally, Table~\ref{table:Notations} represents some of the main notations used throughout the paper.}

The rest of this paper is organized as follows.
In {Section}~\ref{sec:sysmodel}, we describe our location-based system model. A two-phase cache placement scheme comprised of memory allocation and cache arrangement processes is described in {Section}~\ref{sec:cache_placement}, while Section~\ref{sec:delivery} discusses the delivery procedure. In {Section}~\ref{Sec:beamforming}, weighted-max-min beamforming, tailored for the considered location-based cache placement setup, is introduced. In the end, numerical results are provided in {Section}~\ref{sec:Simulations}, while {Section}~\ref{sec:conclusions} concludes the paper.

\begin{table}[ht]
\centering
\smaller
\caption{{Main Notations}}
\begin{tabular}{|l|l|l|l|l|l|}
\hline
{$K$}&  {User count}&  {$\hat{t}$}& {Common global caching gain} &  {$\alpha$}&  {Spatial multiplexing gain}\\ \hline
{$L$} & {Antenna count} & {$Z_k(.)$} & {User $k$ caching function} & {$W(s)$} & {The FoV of state $s$} \\ \hline
 {$y_k$}& {User $k$ received signal} & {$x_{\CU}$} & {Transmitted message to user set $\CU$} & {$\Bv_{\CU}$} & {Dedicated precoder to $x_{\CU}$} \\ \hline
{$T_T$} & {Total delivery time} & {$\Bx_{\bar{\CK}}$} & {Transmitted signal to user set $\bar{\CK}$} & {$\hat{T}_T$} & {Approximated delivery time} \\ \hline
{$\Bh_k$}  & {User $k$ channel vector} & {$r(s)$} & {Approximated throughput at state $s$} & {$\CU_{-k}$} & {All the users in set $\CU \setminus k$} \\ \hline
{$s_k$} & {User $k$ state} & {$R_u,R_w$} & {Approximated state independent rate} & {$m(s)$} & {Dedicated memory to $W(s)$} \\ \hline
{$R_k$}& {User $k$ delivery rate} & {$R^s_u,R^s_w$} & {Approximated symmetric rate} & {$m_k$} & {User $k$ viable memory} \\ \hline
{$W_{d_k}$} & {User $k$ requested file} & {$\varphi_k$} & {User $k$ file-division factor} & {$R^{\CQ}_{\text{sum}}$} & {Sum rate over $|\CQ|$ messages} \\ \hline
{$W^{q}_{\CV(s_k),k}$} & {File segment of $W_{d_k}$} & {$\chi_k$} & {User $k$ file-concatenation factor} & {$D$} & {Desired message count} \\ \hline
 {$W_{\CV(s)}(s)$} & {Sub-file of file $W(s)$}  & {$G_{\CU,k}$} & {Transmitted data to user $k \in \CU$} & {$I$} & {Interfering message count} \\ \hline
{$M$} & {Total available memory} & {$F$} & {File-size} & {$P_T$} & {Total transmit power} \\ \hline
\end{tabular}\label{table:Notations}
\end{table}

\section{System Model}
\label{sec:sysmodel}
We envision a bounded environment (gaming hall, operating theatre, etc.) where a base station (BS) with $L$ transmit antennas serves $K$ single-antenna users\footnote{The system model can be easily extended to multi-antenna receivers following a similar approach as proposed in~\cite{salehi2021lwsa}.} through wireless communication links. The set of users is denoted by $\CK = [K]$. The users are equipped with finite-size cache memories and are free to move throughout the environment. Every user requests data from the BS at each time slot based on its location and the application's needs. {The requested data content can be divided into static and dynamic parts (see Figure~\ref{fig:static_dynamic_decomposition}), and a user needs to obtain both parts to reconstruct the detailed FoV. Typically the major share of the FoV is comprised of the static part, which is the main target in the delivery phase. The dynamic part is delivered parallel to the static part by allocating a portion of the available radio resources (frequency, time, space), depending on the current content demands for both static and dynamic parts. However, in typical virtual gaming scenarios, also the dynamic parts (e.g., geometrical shapes, textures, and avatars) of the FoV are almost entirely cacheable and can be partially stored in advance at the end users~\cite{flashback_2016_VR_static_dynamic_support}. Due to the interaction of the objects in the virtual world, low-overhead control data describing how to reconstruct the dynamic content from both the cached elements and the multicast data must also be provided to the users.} 
This paper {only} focuses on the wireless delivery of the location-dependent {cacheable} content, partially aided by in-device cache memories. A real-world application of such a setup is a wireless XR environment, where the requested data is needed to reconstruct the location-dependent 3D FoV at each user. As a particular example, a 3D XR gaming environment can be considered, where obstacles, walls, buildings, surrounding nature, etc., constitute the static part. On the other hand, players themselves and how they interact with the environment can be considered dynamic content (see Figure~\ref{fig:static_dynamic_decomposition}). Naturally, users located in different locations experience distinct channel conditions due to varying wireless connectivity. Thus, the goal is to design a cache-aided communication scheme that maximizes the achievable rate over the wireless link while also avoiding extensive transmission delays at wireless bottleneck areas.

As discussed in Section~\ref{section:our_contribution}, we assume {that} the application environment is mapped into several STUs, and a separate file is required to construct the scenery at each STU. {Specifically, we assume the requested file contains enough data to render the whole $360$ degree spherical FoV around the user. Any dynamic change prompted by the users' head rotation or other changes in the environment is assumed to be locally rendered by the users, in accordance with the locally available sensory data or after receiving the necessary instruction set from the server for reconstructing and overlaying the dynamic content}. Let us assume that STU mapping is done such that all points in a given STU have almost the same expected level/quality of wireless connectivity. For simplicity, we use the term \textit{state} interchangeably with STU. A graphical representation of a simple application environment with eight states is provided in {Figure}~\ref{fig:system-model}. We use $\CS$ to represent the set of states and assume {that} $|\CS| = S$. Also, the file requested by a user in state $s \in \CS$ is denoted by $W(s)$. Without loss of generality, we assume for every region $s\in\CS$, the size of $W(s)$ is $F$ bits. If not stated otherwise, we consider a normalized data unit in the following and drop $F$ in subsequent notations.

Similar to other centralized coded caching schemes, our new location-dependent scheme works in two distinct phases, I) cache placement and II) content delivery. Each user $k$ is equipped with a cache memory of $M$ (normalized) {data unit} and has a message $Z_{k} = Z_{k}(W(s_{1}), \dots,W(s_{S}))$ stored in its cache during the placement phase, where $Z_{k}(\cdot)$ denotes a function of the files $W(s)$, $\forall s \in \CS$, with entropy not larger than $M$ {data unit}. 

Upon a set of requests $d_k \in \CS$, $\forall k \in \CK$ at the content delivery
phase, the BS multicasts several coded messages, such that at the end of transmission, all users can reliably recover their requested files. Let us assume that coded messages are transmitted in different time intervals and use $\Bx_{\bar{\CK}}$ to denote a coded message that sends data to all the users in $\bar{\CK} \subseteq \CK$.
The number of coded messages and their generation process is detailed in Section~\ref{sec:delivery}. However, as a general description, every message $\Bx_{\bar{\CK}}$ comprises several codewords $x_{\CU}$, where each codeword $x_{\CU}$ contains useful data for a subset of users $\CU \subseteq \bar{\CK}$. Thus, $\Bx_{\bar{\CK}}$ is built as
$    \Bx_{\bar{\CK}} = \sum_{\CU \subseteq \bar{\CK}} \Bv_{\CU}x_{\CU} $,
where $\Bv_{\CU} \in \mathbb{C}^{L}$ denotes the precoding vector dedicated to users in set $\CU$. After {sending} $\Bx_{\bar{\CK}}$, every user $k \in \bar{\CK}$ receives 
\begin{equation}\label{eq:recieved_signal_sysmodel}
    y_{k,\bar{\CK}} = \Bh_{k,\bar{\CK}}\herm\sum_{\CU \subseteq \bar{\CK}}\Bv_{\CU}x_{\CU} + z_k,
\end{equation}
where the channel vector between the BS and user $k$ is denoted by $\Bh_{k,\bar{\CK}} \in \mathbb{C}^{L}$, and $z_k \sim \mathbb{CN}(0,N_0)$ represents the additive white Gaussian noise.
Note that to reproduce the requested file $W_{d_k}$, the decoder of user $k$ makes use of the local cache content $Z_k$ as well as the received signals from the wireless channel over different time intervals (i.e., $y_{k,\bar{\CK}}$). Throughout the rest of the text, we present the delivery procedure for a specific transmission and assume {that} the same procedure is repeated at each transmission. Hence, we use $y_{k}$ and $\Bh_k$ interchangeably with $y_{k,\bar{\CK}}$ and $\Bh_{k,\bar{\CK}}$, respectively.
We assume {that instantaneous and error-free channel state information is available at the transmitter (e.g., via reciprocal reverse link pilot measurements) and used for beamformer design and rate allocation during the delivery phase.} 

{For the exact location-dependent cache placement, we would need to know the normalized achievable throughput $r(s)$ [files/second] at each state $s$. However, this is not possible; since, to compute $r(s)$, we would need prior information about the delivery scheme. This includes, for instance, the number of users scheduled in parallel, all users' locations and channel states, and the precoding algorithms used for data transmission. However, such data is not available during the placement phase. Therefore, expected or approximated delivery rates at each state must be considered for placement purposes. The expected location-specific data rates $r(s)$ can be attained through various means, e.g., via collecting statistics from past active users. In this paper, since the number of users served in parallel in each transmission interval and their respective channel states are not yet known during the placement phase, a hypothetical single-user scenario is considered for the placement to approximate the expected achievable rates $r(s)$.} As a result, the expected interference-free throughput attained in state $s\in \CS$, normalized by the file size $F$ [bits], can be roughly approximated as
\begin{equation} \label{eq: state-rate}
    {r}(s) = \frac{\Omega}{F}C_p\mathbb{E}\big[\log(1+\frac{P_T \ \|\Bh_{k_s}\|^2}{N_0})\big] \quad \text{[files/second]},
\end{equation}
where $C_p$ is a pre-log scaling factor containing any practical overhead, $P_T$ is the transmission power, $\Omega$ is the communication bandwidth, and $\Bh_{k_s} \in \mathbb{C}^{L}$ is the channel vector between the server and a user $k$ located in state $s$.  
Note that the expectation is taken over all user locations and channel realizations in state $s$.  It is worth noting that~\eqref{eq: state-rate} is an `upper bound' for the achievable rate at any state. The main objective is to have a relative throughput measure among states. { The negative throughput scaling due to serving multiple users in parallel (including practical overheads, the impact of scheduling, etc.)} would not change the memory allocation process since it can be considered almost the same across all the states. 

\begin{figure}
    \centering 
    \includegraphics[width=0.5\columnwidth,keepaspectratio]{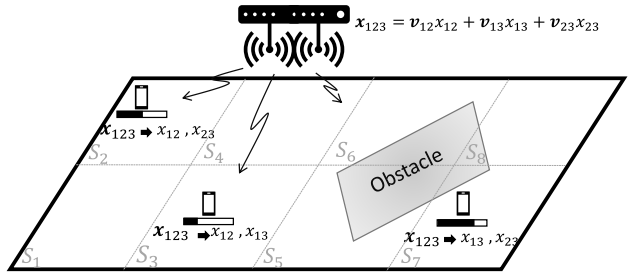}
    \caption{An application environment with $K=3$ users, split into $S=8$ states, where users $k=\{1,2,3\}$ are located in states $s_k=\{3,2,7\}$, respectively. State-specific approximated rates used for cache placement are represented by $r(s)$, where $r(3) > r(2) > r(7)$. The transmitted message ${\bf{x}}_{123}$ consists of  codewords $x_{\mathcal{U}}$, $\CU \in \{ \{1,2\}, \{1,3\}, \{2,3\} \}$, where $x_{\CU}$ contains useful data for users $k \in {\mathcal{U}}$. Beamforming vectors are denoted by $\Bv_{\CU}$. The black bar below each user indicates how much of the requested data is already cached at the user.}
    \label{fig:system-model}
\end{figure}

\section{Location-dependent Cache Placement}
\label{sec:cache_placement}
Different from the existing works, our cache placement phase comprises two consecutive processes, \emph{memory allocation} and \emph{cache arrangement}. The placement phase is executed, for example, before the users enter the application environment or when they pass through specific high data-rate locations (data shower), e.g., nearby the transmitter. During this phase, users' cache memories are proactively filled with valuable data aiming to minimize the required transmission time during the upcoming delivery phase. Note that different from the existing works where the delivery time is optimized only during the content delivery phase, we proactively consider minimizing the delivery time also in the placement phase. As a result, the contents relevant to locations with poor wireless connectivity are prioritized in the user's memories to help prevent excessive delays during upcoming transmissions.

\noindent\textbf{Memory Allocation:}
Due to the considered real-time application, it is crucial to guarantee to deliver the requested data within a limited time. Intuitively, this requires reserving a larger share of the total cache memory for storing data needed in locations with poor communication quality. In this regard, the amount of cache memory dedicated for storing (parts of) every state-specific content file $W(s)$ at each user is determined during the \textit{memory allocation} process. In this paper, we assume that there is no {a priori} knowledge about the users' spatial locations during the {placement} phase. Hence, for memory allocation, we consider uniform access probability for all the states (using prior knowledge about the states' popularity, the performance can be further improved). Let us use $m(s)$ to denote the normalized cache size at each user allocated to store (parts of) $W(s)$. Since the size of $W(s)$ is normalized to one, a user in state $s$ needs to receive $1-m(s)$ data units over the wireless link to reconstruct the FoV of state $s$. {Intuitively, the delivery time for a user in state $s$ can be approximated by $\frac{1-m(s)}{r(s)}$ in the single-user case, where $r(s)$ is the approximated rate at state $s$ (c.f. Eq.~\eqref{eq: state-rate}). However, for the multi-user case, the approximate delivery time will be somewhat different. The approximated delivery time $\hat{T}_T$ when multiple users are served in parallel will be detailed in Sec IV, where we show that if $m(s)$ values are known, $\hat{T}_T$ is formulated as
\begin{equation} \label{eq: aprrox_deliv_time}
    \hat{T}_T = \frac{K}{\bar{t}+{\alpha}}\max_{s \in \CS}\frac{1-m(s)}{r(s)} \quad \textrm{[seconds]},
\end{equation}
where $\bar{t} = K\min\limits_{s \in [S]}m(s)$ is the minimum achievable \textit{global caching gain} given the non-uniform memory allocation. {We use $\alpha \leq L$ to denote the spatial multiplexing gain, which can be tuned for a given scenario based on, e.g., available transmit power, constraints on the beamformer design, the number of users in the network, etc. (c.f.~\cite{tolli2017multi})}.}

{Note that the $\bar{t}+{\alpha}$ term in the denominator represents a lower bound on the achievable DoF for the non-uniform memory allocation scenario (for the uniform allocation, the DoF of $K\frac{M}{S}+{\alpha}$ is achievable~\cite{tolli2017multi}). The term $K \max\limits_{s \in \CS}\frac{1-m(s)}{r(s)}$ approximates the worst-case delivery time across all the states when $K$ users are served simultaneously.} {Next}, we first rewrite~\eqref{eq: aprrox_deliv_time} as $\hat{T}_T = \frac{1}{\bar{m}+\frac{{\alpha}}{K}}\max_{s \in \CS}\frac{1-m(s)}{r(s)}$, where $\bar{m}$ equals $\min\limits_{s \in [S]}m(s)$. Then, to minimize the delivery time for any possible realizations of user locations, we formulate the memory allocation process as the following linear fractional programming (LFP) problem:
\begin{equation} 
\begin{aligned}
\label{cache-allocation}
&\min_{m(s), \; \gamma \ge 0, \; {m} \ge 0} \quad  \frac{\gamma}{{m}+\frac{{\alpha}}{K}}\\
&\textrm{s.t.} \quad  \frac{1-m(s)}{r(s)} \leq \gamma, \ \forall s \in \CS, \\ &{m} \leq m(s), \ \forall s \in \CS,  \quad \sum_{s \in \CS} m(s) \leq M.
\end{aligned}
\end{equation}
Note that at the optimal point, $ m = \bar{m} = \min\limits_{s \in [S]}m(s)$. Using Charnes-Cooper transformation~\cite{charnes1962programming}, this problem can be reformulated as an equivalent linear program (LP) as follows: 
\begin{equation} 
\begin{aligned}
\label{cache-allocation_dif_conv}
&\min_{m^{'}(s), \; \gamma^{'} \ge 0, \; {m}^{'} \ge 0 {,} \xi \ge 0} \quad  \gamma^{'}\\
&\textrm{s.t.} \quad  \frac{\xi-m^{'}(s)}{r(s)} \leq \gamma^{'}, \ \forall s \in \CS, \\ &\bar{m}^{'} \leq m^{'}(s), \ \forall s \in \CS, \quad \sum_{s \in \CS} m^{'}(s) \leq M\xi,
\quad {m}^{'} + \frac{{\alpha}}{K}\xi = 1.
\end{aligned}
\end{equation}
Note that after solving this problem, the actual allocated memory would be $m(s) = m^{'}(s)/\xi, \forall s$, which also implies ${m} = {m}^{'}/\xi$.

\begin{remark}
Substituting the term $\frac{{\alpha}}{K}$ in~\eqref{cache-allocation} with a general constant term $\phi$ enables a trade-off between the local and global caching gains. Selecting a large $\phi \gg \frac{{\alpha}}{K}$ prioritizes the local caching gain $m(s)$ at the expense of the minimum global caching gain $\bar{t}=K\bar{m}$ (as the denominator in the objective function becomes almost constant).
On the other hand, if $\phi \ll \frac{{\alpha}}{K}$, the minimum allocated memory $\bar{m}$ converges to $\frac{M}{S}$ and the minimum global caching gain is maximized at the cost of lower local caching gain for the states with poor expected connectivity $r(s)$, resulting in higher {delivery time} fluctuations.
\end{remark}

\noindent\textbf{Cache Arrangement:}
After the memory allocation process, we store data fragments in the cache memories of the users following a similar method as proposed in~\cite{MaddahAli-2014}. To this end, for every state $s \in \CS$, we first split $W(s)$ into~$\binom{K}{t(s)}$ sub-files denoted by $W_{\CV(s)}(s)$, where $t(s) = K m(s)$ and $\CV(s)$ can be any subset of the user-set $\CK$ with $|\CV(s)| = t(s)$. Then, at the cache memory of user $k \in \CK$, we store $W_{\CV(s)}(s)$ for every state $s \in \CS$ and every subset $\CV(s) \ni k$. The cache arrangement process is outlined in Algorithm~\ref{Alg:placement}. For simplicity, here we assume that for every $s\in\CS$, $m(s)>0$, and $t(s)$ is an integer (the general case where these assumptions are not necessarily met is discussed in Appendix~\ref{sec: Appendix A}). Also, for notational simplicity, we ignore the brackets and commas while explicitly referring to a given $\CV(s)$, e.g., $W_{ij}(s) \equiv W_{\{i,j\}}(s)$. 

\begin{algorithm}[t]
	\caption{Location-based cache placement}
	\begin{algorithmic}[1]
		\Procedure{CACHE\_PLACEMENT}{}
		\State solve for $\{m(s) \}$ with~\eqref{cache-allocation}
		
		\ForAll{$s \in \CS$} 
		
		\State $t(s) = K \times m(s)$ 
		
		\State $W(s) \rightarrow \{W_{\CV(s)}(s) \; | \; \CV(s) \subseteq \CK, |\CV(s)| = t(s)\}$
		
		\ForAll{$\CV(s)$}
		\ForAll{$k \in \CK$}
		\If{$k \in \CV(s)$}
		    \State Store $W_{\CV(s)}(s)$ in the cache of user $k$
		\EndIf
		\EndFor
		\EndFor
		\EndFor
		\EndProcedure 
	\end{algorithmic}
	\label{Alg:placement}
\end{algorithm}

\begin{exmp}
\label{exmp:placement}
Consider a simplified XR application scenario with $K=4$ users, where the application area is split into $S=5$ states, and for each state, the required data size is $F=400 \textrm{MB}$. Each user has a cache size of $900\textrm{MB}$; hence, the normalized cache size is $M = 2.25$ data units. Assume {that} the spatial distribution of the approximated normalized throughput is as given in Table~\ref{Table:rate&cache_distribution}, where the memory allocations resulting from solving~\eqref{cache-allocation} are also shown. 
It can be easily verified that $t(1) = t(5) = 1$, $t(2) = t(4) = 2$, and $t(3) = 3$. As a result, $W(1)$, $W(3)$ and $W(5)$ should be split into $4$ sub-files, while $W(2)$ and $W(4)$ are split into $\binom{4}{2}=6$ sub-files. The resulting cache placement is visualized in Figure~\ref{fig:cache pool}.

\begin{table}[t]
\centering
\begin{tabular}{|c||c|c|c|c|c|}
\cline{2-6}
 \multicolumn{1}{c|}{} & $s\!=\!1$   & $s\!=\!2$   & $s\!=\!3$   & $s\!=\!4$   & $s\!=\!5$    \\ \hline
$r(s)$                               & $3\times 10^3$    & $2\times 10^3$   & $1\times 10^3$   & $2\times 10^3$   & $3\times 10^3$    \\ \hline
$m(s)$                               & 0.25 & 0.5 & 0.75 & 0.5 & 0.25 \\ \hline
\end{tabular}
\caption{Location-specific rate and memory allocation for Example~\ref{exmp:placement}.}
\label{Table:rate&cache_distribution}
\end{table}

\begin{figure}[ht]
    \Large
    \centering
    \resizebox{\columnwidth}{!}{%
    \begin{tabular}{c}
        \begin{tabular}{|c|c|c|c|c|c|c|c|c|c|c|c|c|c|}
            \cline{2-13}
            \multicolumn{1}{c|}{} & \multicolumn{3}{c|}{$\,\quad \quad W_1(s) \ \quad$} & \multicolumn{3}{c|}{$\,\quad \quad W_2(s) \ \quad$} & \multicolumn{3}{c|}{$\,\ \quad W_3(s) \ \quad$} & \multicolumn{3}{c|}{$\,\ \quad W_4(s) \ \quad$} \\
            \hline
            \multirow{4}{*}{\rotatebox[origin=c]{90}{ $s=1,5$}} & \multicolumn{3}{c|}{\cellcolor{gray!50}} & 
            \multicolumn{3}{c|}{} & 
            \multicolumn{3}{c|}{} & 
            \multicolumn{3}{c|}{} & user $1$ \\
            \cline{2-14}
            & \multicolumn{3}{c|}{} & 
            \multicolumn{3}{c|}{\cellcolor{gray!50}} & 
            \multicolumn{3}{c|}{} & 
            \multicolumn{3}{c|}{} & user $2$ \\
            \cline{2-14}
            & \multicolumn{3}{c|}{} & 
            \multicolumn{3}{c|}{} & 
            \multicolumn{3}{c|}{\cellcolor{gray!50}} & 
            \multicolumn{3}{c|}{} & user $3$ \\
            \cline{2-14}
            & \multicolumn{3}{c|}{} & 
            \multicolumn{3}{c|}{} & 
            \multicolumn{3}{c|}{} & 
            \multicolumn{3}{c|}{\cellcolor{gray!50}} & user $4$ \\
            \hline
            \multicolumn{14}{c}{}\\[-1.6em]
        \end{tabular} 
        \begin{tabular}{|c|c|c|c|c|c|c|c|c|c|c|c|c|c|}
            \cline{2-13}
            \multicolumn{1}{c|}{} & \multicolumn{2}{c|}{$W_{12}(s)$} & \multicolumn{2}{c|}{$W_{13}(s)$} & \multicolumn{2}{c|}{$W_{14}(s)$} & \multicolumn{2}{c|}{$W_{23}(s)$} & \multicolumn{2}{c|}{$W_{24}(s)$} & \multicolumn{2}{c|}{$W_{34}(s)$}  \\
            \hline
            \multirow{4}{*}{\rotatebox[origin=c]{90}{ $s=2,4$}} & \multicolumn{2}{c|}{\cellcolor{gray!50}} & 
            \multicolumn{2}{c|}{\cellcolor{gray!50}} & 
            \multicolumn{2}{c|}{\cellcolor{gray!50}} & 
            \multicolumn{2}{c|}{} & \multicolumn{2}{c|}{} & \multicolumn{2}{c|}{} & user $1$ \\
            \cline{2-14}
            & \multicolumn{2}{c|}{\cellcolor{gray!50}} & 
            \multicolumn{2}{c|}{} & 
            \multicolumn{2}{c|}{} & 
            \multicolumn{2}{c|}{\cellcolor{gray!50}} & \multicolumn{2}{c|}{\cellcolor{gray!50}} & \multicolumn{2}{c|}{} & user $2$ \\
            \cline{2-14}
            & \multicolumn{2}{c|}{} & 
            \multicolumn{2}{c|}{\cellcolor{gray!50}} & 
            \multicolumn{2}{c|}{} & 
            \multicolumn{2}{c|}{\cellcolor{gray!50}} & \multicolumn{2}{c|}{} & \multicolumn{2}{c|}{\cellcolor{gray!50}} & user $3$ \\
            \cline{2-14}
            & \multicolumn{2}{c|}{} & 
            \multicolumn{2}{c|}{} & 
            \multicolumn{2}{c|}{\cellcolor{gray!50}} & 
            \multicolumn{2}{c|}{} & \multicolumn{2}{c|}{\cellcolor{gray!50}} & \multicolumn{2}{c|}{\cellcolor{gray!50}} & user $4$ \\
            \hline
            \multicolumn{14}{c}{}\\[-1.6em]
        \end{tabular}  
        \begin{tabular}{|c|c|c|c|c|c|c|c|c|c|c|c|c|c|}
            \cline{2-13}
            \multicolumn{1}{c|}{} & \multicolumn{3}{c|}{$\;\ \ W_{123}(s) \,\quad$} & \multicolumn{3}{|c}{$\;\ \ W_{124}(s) \,\quad$} & \multicolumn{3}{|c}{$\;\ \ W_{134}(s) \,\quad$} & \multicolumn{3}{|c|}{$\;\ \ W_{234}(s) \,\quad$} \\
            \hline
            \multirow{4}{*}{\rotatebox[origin=c]{90}{ $s=3$}} & \multicolumn{3}{c|}{\cellcolor{gray!50}} & 
            \multicolumn{3}{c|}{\cellcolor{gray!50}} & 
            \multicolumn{3}{c|}{\cellcolor{gray!50}} & 
            \multicolumn{3}{c|}{} & user $1$ \\
            \cline{2-14}
            & \multicolumn{3}{c|}{\cellcolor{gray!50}} & 
            \multicolumn{3}{c|}{\cellcolor{gray!50}} & 
            \multicolumn{3}{c|}{} & 
            \multicolumn{3}{c|}{\cellcolor{gray!50}} & user $2$ \\
            \cline{2-14}
            & \multicolumn{3}{c|}{\cellcolor{gray!50}} & 
            \multicolumn{3}{c|}{} & 
            \multicolumn{3}{c|}{\cellcolor{gray!50}} & 
            \multicolumn{3}{c|}{\cellcolor{gray!50}} & user $3$ \\
            \cline{2-14}
            & \multicolumn{3}{c|}{} & 
            \multicolumn{3}{c|}{\cellcolor{gray!50}} & 
            \multicolumn{3}{c|}{\cellcolor{gray!50}} & 
            \multicolumn{3}{c|}{\cellcolor{gray!50}} & user $4$ \\
            \hline
        \end{tabular}
    \end{tabular}%
    }
    \caption{Cache placement visualization for Example~\ref{exmp:placement}.}
    \label{fig:cache pool}
\end{figure}

\end{exmp}

{To show that the memory constraint is strictly satisfied, we remind that each $W(s)$ is split into $\binom{K}{t(s)}$ subfiles $W_{\mathcal{V}(s)}(s)$, where $\mathcal{V}(s)$ can be any subset of users with size $t(s)$. Then, each user~$k$ stores every $W_{\mathcal{V}(s)}(s)$ for which $k \in \mathcal{V}(s)$. In other words, $W(s)$ is split into $\binom{K}{t(s)}$ sub-files, from which $\binom{K-1}{t(s)-1}$ sub-files are stored in the cache memory of each user.} Hence, the total memory size dedicated to $W(s)$ at each user is
\begin{equation}
    \frac{\binom{K-1}{t(s)-1}}{\binom{K}{t(s)}} = \frac{t(s)}{K} = m(s) \; ,
\end{equation} 
i.e., the proposed algorithm satisfies the cache size constraints. 

\section{Asymmetric Cache-aided Content Delivery}
\label{sec:delivery}
At the beginning of the delivery phase, every user $k \in \CK$ reveals its requested file $W_{d_k} \equiv W(s_k)$. Note that, according to the system model, $W_{d_k}$ depends on the state $s_k$ where user $k$ is located. The server then builds and transmits several \textit{nested} codewords,\footnote{In this paper, we consider the NCM scheme~\cite{tang2011full} to support multi-rate transmission. However, the scheme is oblivious to the modulation procedure, and any other multi-rate modulation scheme could be used (e.g.,~\cite{kramer_broadcast_channel_side_information_ITW2007, chen2010novel, tang2011full,asadi_side_information_broadcast_multi_rate_TIT2015}).} such that after receiving the codewords, all the users can reconstruct their requested files. As detailed in Section~\ref{sec:sysmodel}, user $k$ requires a total amount of one normalized data unit to reconstruct $W_{d_k}$. However, only a subset of this data, with size $m_k \equiv m(s_k)$, is available in its cache, and the remaining part should be delivered by the server. Note that the conventional multi-server \ac{CC}-based delivery schemes (e.g.,~\cite{shariatpanahi2018physical} and~\cite{tolli2017multi}) that assume all users cache the same amount of data do \textit{not} apply to our considered scenario where each user has cached a different amount of its requested file. Thus, a new delivery mechanism is required to achieve a proper multicasting gain.

The new delivery algorithm is outlined in Algorithm~\ref{Alg:Delivery}. First, the server builds and transmits multiple transmission vectors $\mathbf{x}_{\Bar{\CK}}$ in a \ac{TDMA} manner for every subset of users $\Bar{\CK}  \subseteq \CK : |\Bar{\CK}| = \hat{t} + {\alpha}$, where $\hat{t} = \min_{k \in \CK} t_k$ is the common global caching gain and $t_k \equiv t(s_k)$.\footnote{Note that in the placement phase (Section~\ref{sec:cache_placement}), we used $\bar{t} =  \min_{s \in \CS} t(s)$, i.e., the minimum was taken over all the states. However, during the delivery phase where the actual locations of the users are known, we use $\hat{t} =  \min_{k \in \CK} t_k$ to take the minimum over all the users' locations.} The transmitted signal vector 
\begin{equation}\label{eq: multiplexed signal} 
    \Bx_{\bar{\CK}} = \sum_{\CU \subseteq \bar{\CK}, |\CU| = \hat{t}+1}\Bv_{\CU}x_{\CU} \; , 
\end{equation}
is comprised of multiple nested codewords $x_{\CU}$, where $\CU$ can be any subset of $\Bar{\CK}$ with $|\CU| = \hat{t} + 1$. {The elements (constellation points) of every nested codeword $x_{\mathcal{U}}$ are drawn from complex Gaussian distribution such that $\mathbb{E}[|x_{\mathcal{U}}|^2]=1$.  The details of the nesting operation, as well as the coding and decoding procedure, are explained in~\cite[Section 4]{tang2011full}.} Also, every $x_{\CU}$ is precoded with a tailored beamformer vector $\Bv_{\CU} \in \mathbb{C}^{L}$, designed to suppress (or null-out) the interference caused by $x_{\CU}$ on every user in $\bar{\CK} \setminus \CU$. After the transmission of $\Bx_{\bar{\CK}}$, the corresponding received signal at user $k \in \bar{\CK}$ follows equation~\eqref{eq:recieved_signal_sysmodel}. 

\begin{algorithm}[t]
    \caption{NCM-based Content Delivery}
	\begin{algorithmic}[1]
		\Procedure{DELIVERY}{}
		
		\State $\hat{t} = \min_{k \in \CK} t_k$
		\ForAll{$\Bar{\CK} \subseteq \CK : |\Bar{\CK}| = \hat{t} + {\alpha}$} \label{alg:delivery_K_subsets}
		\State $\Bx_{\Bar{\CK}} \gets 0$
		    \ForAll{$\CU \subseteq \Bar{\CK} : |\CU| = \hat{t}+1$}\label{alg:delivery_U_subsets}
		        \State $x_{\CU} \gets 0$
		        \ForAll{$k \in \CU$}
		            \State ${\varphi_k} \gets \binom{t_k}{\hat{t}}\binom{K-\hat{t}-1}{{\alpha}-1}$, ${G_{\CU,k}} \gets 0$, $\CU_{-k} \gets \CU \backslash \{k \}$
		            \ForAll{$\CV_k \subseteq \CK : |\CV_k| = t_k$}
		                \If{$\CU_{-k} \subseteq \CV_k$, $k \not\in \CV_k$} \label{alg:delivery_tK_subsets}
		                \State $W_{\CV_k,k}^q \gets$ \textsc{Chunk}($W_{\CV_k,k}, {\varphi_k}$)
		                \State ${G_{\CU,k}} \gets$ \textsc{Concat} $ ({G_{\CU,k}}, W_{\CV(s_k),k}^q)$ 
		            \EndIf
		            \EndFor
		        \State $x_{\CU} \gets$ \textsc{Nest} $(x_{\CU}, {G_{\CU,k}})$ 
		        \EndFor
		        \State $\Bx_{\Bar{\CK}} \gets \Bx_{\Bar{\CK}} + \Bv_{\CU}x_{\CU}$
		    \EndFor
		    \State Transmit $\Bx_{\Bar{\CK}}$
		\EndFor
		\EndProcedure
	\end{algorithmic}
	\label{Alg:Delivery}
\end{algorithm}

The nested codeword $x_{\CU}$ is built to include a useful data term/packet ${G_{\CU,k}}$ for every user $k \in \CU$, {i.e., $x_{\CU} = \underset{k \in \CU}{*}\big(G_{\CU,k}\big)$, where $(*)$ denotes the nesting operation (c.f.,~\cite{tang2011full}}). The data term ${G_{\CU,k}}$ is chosen to be available in the cache memory of every other user in $\CU \setminus \{k\}$, so that these users can remove its interference using their cache contents. To satisfy this condition, denoting $\CU_{-k} \equiv \CU \setminus \{k\}$, we build ${G_{\CU,k}}$ to include (parts of) every \textit{suitable} sub-file $W_{\CV(s_k),k}$ for which $\CU_{-k} \subseteq \CV(s_k)$ and $k \not\in \CV(s_k)$. However, since $W_{\CV(s_k),k}$ is cached in the cache memory of every user in $\CV(s_k)$, and also because $|\CU_{-k}| = \hat{t} \le t_k = |\CV(s_k)|$,  we may find more than one suitable sub-file $W_{\CV(s_k),k}$ to be included in ${G_{\CU,k}}$. In fact, there exists exactly 
\begin{equation*} 
    \chi_k = \binom{K-\hat{t}-1}{t_k-\hat{t}}
\end{equation*}
suitable sub-files for inclusion in ${G_{\CU,k}}$, which should be split into smaller parts and \textit{concatenated} while building $x_{\CU}$.
Note that every sub-file $W_{\CV(s_k),k}$ appears in $\binom{t_k}{\hat{t}}$ different $\CU_{-k}$ sets (c.f. step~\ref{alg:delivery_tK_subsets} in Algorithm~\ref{Alg:Delivery}), and each user set $\CU$ is targeted $\binom{K-\hat{t}-1}{{\alpha}-1}$ times during the delivery phase (c.f. steps~\ref{alg:delivery_K_subsets} and~\ref{alg:delivery_U_subsets}).
Hence, to send fresh content in each transmission, we need to divide every sub-file $W_{\CV(s_k),k}$ suitable for user $k$ into exactly
\begin{equation*} 
    {\varphi_k} = \binom{t_k}{\hat{t}}\binom{K-\hat{t}-1}{{\alpha}-1}
\end{equation*}
equal-sized segments (denoted by $W_{\CV(s_k),k}^q$ in Algorithm~\ref{Alg:Delivery}) before the concatenation. In other words, we split every suitable sub-file into ${\varphi_k}$ segments, and then concatenate $\chi_k$ number of these segments to build ${G_{\CU,k}}$, {i.e., $G_{\CU,k} = \underset{\substack{\CV(s_k) \subseteq \CK\setminus k, \\ \CU_{-k} \subseteq \CV(s_k)}}{\prod}\big(W^{q}_{\CV(s_k),k}\big)$, where $\prod(A,B)$ is the bitwise concatenation of files $A$ and $B$}. 

The function $\textsc{Chunk}$ in Algorithm~\ref{Alg:Delivery} ensures none of the segments of a sub-file is sent twice, and {functions $\textsc{Concat}$ and $\textsc{Nest}$ denote bit-wise concatenation ($\prod$) and nesting operation $(*)$, respectively.} We will later discuss in section~\ref{Sec:beamforming} that using the nesting operation (c.f.~\cite{tang2011full}) to create codeword $x_{\CU}$, we can simultaneously transmit every ${G_{\CU,k}}$ with rate $R_k$ such that $\frac{|{G_{\CU,k}}|}{R_k} = \frac{|{G_{\CU,i}}|}{R_i}, \forall(k,i) \in \CU$.

\begin{exmp}
\label{exmp:interference-free-delivery}
Consider the network in Example~\ref{exmp:placement}, for which the cache placement is visualized in Figure~\ref{fig:cache pool}. Assume {that} there exist two antennas at the transmitter (i.e., $L = {\alpha} =2$). Let us consider a specific time slot, in which $s_1 = 1$, $s_2 =2$, $s_3 =4$, $s_4 = 5$. Denoting the set of requested sub-files for user $k$ with $\CT_k$ and assuming $A \equiv W(1)$, $B \equiv W(2)$, $C \equiv W(4)$, and $D \equiv W(5)$, we have
\begin{equation}  \small
    \begin{aligned}
    \CT_1 = \{ A_{2}, A_{3}, A_{4} \}, \qquad
    \CT_2 = \{ B_{13}, B_{14}, B_{34} \}, \quad
    \CT_3 = \{ C_{12}, C_{14}, C_{24} \}, \quad
    \CT_4 = \{ D_{1}, D_{2}, D_{3} \}. 
    \end{aligned}
\end{equation}
Note that the size of the sub-files of $A,B,C,D$ are $\frac{1}{4}$, $\frac{1}{6}$, $\frac{1}{6}$, and $\frac{1}{4}$ data units, respectively. As $L=2$ and the common global caching gain is $\hat{t} = 1$, our proposed algorithm can deliver data to $\hat{t}+{\alpha} = 3$ users during each transmission. Let us consider the transmission vector $ \Bx_{123}$ for users $\Bar{\CK} = \{1,2,3\}$. 
Following equation~\eqref{eq: multiplexed signal}, we have $\Bx_{123} = \Bv_{12}x_{12} +\Bv_{13}x_{13} +\Bv_{23}x_{23}$, where the nested codewords $x_{12}$, $x_{13}$, and $x_{23}$ deliver a portion of the requested data to user sets $\{1,2\}$, $\{1,3\}$ and $\{2,3\}$, respectively. Based on the users' request sets $\CT_k$, for each of the nested codewords $x_{12}$ and $x_{13}$, there exists only one suitable sub-file for user~$1$ (i.e., $ \chi_1 = 1$), and these sub-files should be split into $\alpha_1 = 2$ segments. However, for users~2 and~3, we have $\chi_2 = \chi_3 = 2$ and the segmentation factor is $\alpha_2 = \alpha_3 = 4$. As a result, $x_{12}$ is built as $x_{12} = A_{2}^{1} * \prod(B_{13}^{1}, B_{14}^{1}),$ where 
superscripts are used to differentiate various segments of a sub-file. The nesting operation in $x_{12}$ is performed such that $A_2^{1}$ and $\prod(B_{13}^{1}, B_{14}^{1})$ are delivered with proportional rates $R_1 = \frac{3}{2}*R_2$. Following the same procedure to build $x_{13}$ and $x_{23}$, the transmission vector $ \Bx_{123}$ is formed as
\begin{equation}\nonumber  \small
\begin{aligned}
    \Bx_{123} =& \Bv_{12} \! \left(A_2^{1} *\prod(B_{13}^{1},B_{14}^{1})\right) + \Bv_{13} \! \left(A_3^{1} *\prod(C_{12}^{1}, C_{14}^{1})\right) +\Bv_{23} \! \left(\prod(B_{13}^{2},B_{34}^{1}) *\prod(C_{12}^{2},C_{24}^{1})\right).
\end{aligned}
\end{equation}

Now, let us consider the decoding process for $\Bx_{123}$ at user~$1$. Following Eq.~\eqref{eq:recieved_signal_sysmodel}, user~$1$ receives
\begin{equation}\nonumber \small
\begin{aligned}
        y_1 &=  \underline{A_2^{1} *\prod(B_{13}^{1},B_{14}^{1})}\Bh_1^H\Bv_{12} + \underline{A_3^{1} *\prod(C_{12}^{1},C_{14}^{1})}\Bh_1^H\Bv_{13} + w_1,
\end{aligned}
\end{equation}
where the term $w_1 \equiv \Bv_{23} \left(\prod(B_{13}^{2},B_{34}^{1}) *\prod(C_{12}^{2},C_{24}^{1})\right) + z_1$ contains the (suppressed) interference terms and noise at user~1. Now, to recover its requested data terms $A_{2}^1$ and $A_{3}^1$, user~$1$ has to jointly decode the two underlined messages (using successive interference cancellation (SIC), c.f.~\cite{tolli2017multi}), benefiting from its cache contents (i.e., $\prod(B_{13}^{1},B_{14}^{1})$ and $\prod(C_{12}^{1},C_{14}^{1})$) as a priori knowledge for demodulation. {Specifically, to achieve a symmetric rate over two messages $A^1_2$ and $A_3^1$, a SIC receiver combined with appropriate time-sharing between different decoding orders in the resulting MAC region is required~\cite{tolli2017multi}. For example, the server first allocates rates such that user $1$ is able to first decode $A^1_2$ assuming $A_3^1$ as interference, then remove $A^1_2$ from $y_1$ and finally decode $A_3^1$ interference-free. Then, for another time interval, the rate allocation changes such that user $1$ decodes $A_3^1$ first and $A^1_2$ last.} Similarly, users~2 and~3 can also decode $\{B_{13}^{1}, \ B_{13}^{2}, \ B_{14}^{1}, \ B_{34}^{1}\}$ and $\{C^1_{12}$, $C^2_{12}$ $C^1_{14}$, $C^1_{24}\}$ from $X_{123}$, respectively. Following algorithm~\ref{Alg:Delivery}, {the other three transmissions to deliver all the missing parts are} 
\begin{equation*} \small
    \begin{aligned}
        \Bx_{124} =& \Bv_{12} \left(A_2^{2} *\prod(B_{13}^{3}, B_{14}^{2})\right) + \Bv_{14} \left(A_4^{1} * D_{1}^{1}\right) + \Bv_{24} \left(\prod(B_{14}^{3}, B_{34}^{2})*D_{2}^{1}\right),  \\
        \Bx_{134} =& \Bv_{13} \left(A_3^{2} *\prod(C_{12}^{3}, C_{14}^{2})\right) + \Bv_{14} \left(A_4^{2} *D_{1}^{2}\right) + \Bv_{34} \left(D_{3}^{1}*\prod(C_{14}^{3}, C_{24}^{2})\right),  \\
        \Bx_{234} =& \Bv_{23} \left(\prod(B_{13}^{4}, B_{34}^{4}) *\prod(C_{12}^{4}, C_{24}^{3})\right) + \Bv_{24} \left(D_2^{2} *\prod(B_{14}^{4}, B_{34}^{3})\right) +  \Bv_{34} \left(D_{3}^{2}*\prod(C_{14}^{4}, C_{24}^{4})\right).  
    \end{aligned}
\end{equation*}
Note that compared to~\cite{tolli2017multi}, in all these transmissions, we serve users~$1$ and~$4$ with a higher rate ($1.5$ times) compared to users~$2$ and~$3$ using the NCM modulation~\cite{tang2011full}.
\end{exmp}

\begin{lemma}\label{theorm: delivery}
Using the proposed cache placement and content delivery algorithms, every user receives its requested data.
\end{lemma}
\begin{proof}
The user $k$ in state $s_k$ needs to receive $1-m_k$ data units during the delivery phase. This data is delivered by $\binom{K-1}{\hat{t}+{\alpha}-1}$ transmission vectors $\Bx_{\bar{\CK}}$ for which $k \in \bar{\CK} $, i.e., by all the user subsets $\bar{\CK}$ that include user $k$. However, the number of nested codewords $x_{\CU}$ for which $k \in \CU$ in every such vector $\Bx_{\bar{\CK}}$ is $\binom{\hat{t}+{\alpha}-1}{\hat{t}}$, and each $x_{\CU}$ delivers to user $k$ a data term ${G_{\CU,k}}$ that is comprised of $\chi_k$ segments each with size ${1}/{{\varphi_k} \binom{K}{t_k}}$ data units. Hence, the total data size delivered to user $k$ is  
$\frac{ \binom{K-1}{\hat{t}+{\alpha}-1}\binom{\hat{t}+{\alpha}-1}{\hat{t}} \binom{K-\hat{t}-1}{t_k - \hat{t}} }{ \binom{K}{t_k} \binom{t_k}{\hat{t}}  \binom{K-\hat{t}-1}{{\alpha} - 1}} = \frac{K-t_k}{K} = 1 - m_{k} .$
\end{proof}

\subsection{Weighted Max-Min Beamforming} \label{Sec:beamforming}

In this section, we illustrate how the beamforming vectors $\Bv_{\CU}$ in~\eqref{eq: multiplexed signal} are built. Note that, due to the underlying multi-rate transmission requirement of our proposed scheme, the optimized beamformer design in~\cite{tolli2017multi} is not readily applicable here. Thus, unlike~\cite{tolli2017multi} that considers max-min-fairness to design precoders, here we formulate the objective function as a weighted-max-min (WMM) problem, where the weights reflect the non-uniform amounts of data transmitted to different users. 

As discussed in Section~\ref{sec:delivery}, for the proposed scheme, data delivery is done using $\binom{K}{\hat{t}+{\alpha}}$ transmission vectors $\Bx_{\bar{\CK}}$. Also, every vector $\Bx_{\bar{\CK}}$ comprises $\binom{\hat{t}+{\alpha}}{\hat{t}+1}$ data terms $x_{\CU}$ and the same number of beamforming vectors $\Bv_{\CU}$, as shown in~\eqref{eq: multiplexed signal}.  As a result, after the transmission of $\Bx_{\bar{\CK}}$, the received signal in~\eqref{eq:recieved_signal_sysmodel} can be rewritten as
\begin{equation} \label{eq:received_signal} 
    y_{k} = \sum_{\CU \subseteq \bar{\CK}, \CU \ni k, |\CU| = \hat{t}+1}\Bh_{k}^{H}\underline{\Bv_{\CU}x_{\CU}} + \sum_{\CV \subseteq \bar{\CK} \setminus k, |\CV| = \hat{t}+1}\Bh_{k}^{H}\Bv_{\CV}x_{\CV} + z_k,
\end{equation}
where each of the $D = \binom{\hat{t}+{\alpha}-1}{\hat{t}}$ underlined terms in~\eqref{eq:received_signal} contains fresh data  for user $k$ with size 
\begin{equation} 
    c_k = \frac{\binom{K-\hat{t}-1}{t_k - \hat{t}} }{ \binom{K}{t_k} \binom{t_k}{\hat{t}}  \binom{K-\hat{t}-1}{{\alpha} - 1}} = \frac{1-m_k}{\binom{K-1}{\hat{t}+{\alpha}-1}\binom{\hat{t}+{\alpha}-1}{\hat{t}}},
\end{equation}
and the rest $I = \binom{\hat{t}+{\alpha}-1}{\hat{t+1}}$ non-underlined terms are seen as interference. Thus, from the user $k$'s perspective, $y_k$ is a \ac{MAC} with $D$ desired messages and $I$ interference terms. Let us use $\CD_k = \{\CU \ | \ \CU \subseteq \bar{\CK}, |\CU| = \hat{t}+1, \CU \ni k\}$ to denote the set of all the desired message indices for user $k$, i.e., $|\CD_k| = D$. To minimize the overall time to decode all the $D$ desired messages, they should all be transmitted with the same rate $ R_{k}$ , i.e.,
\begin{equation}\label{eq:MAC_rate}
    R_{k} = \min_{\CQ \subseteq \CD_k, |\CQ| \neq 0} \frac{1}{|\CQ|}R_{\text{sum}}^{\CQ},
\end{equation}
where $R_{\text{sum}}^{\CQ} = \log\left(1 + \frac{\sum_{\CU \in \CQ}|\Bh_{k}^{H}\Bv_{\CU}|^2}{\sum_{\CV \in \CI_k}|\Bh_{k}^{H}\Bv_{\CV}|^2 + \sigma^2}\right)$ denotes the sum rate over all $|\CQ|$ messages, and $\CI_k := \{\CV \ | \ \CV \subseteq \bar{\CK} \setminus k, |\CV| = \hat{t}+1\}$ is the set of interfering message indices for user $k$. Note that $R_{k}$ is the symmetric rate per message, and since user $k$ receives $D$ messages in each transmission, its overall symmetric rate would be $D R_{k}$. Moreover, since the total size of received data at this user is $Dc_k$, the required delivery time for user $k$ is $T_k = \frac{Dc_k}{D R_{k}} = \frac{c_k}{ R_{k}}$, and the delivery time for $\Bx_{\bar{\CK}}$ would be $T_{\bar{\CK}} =  \max\limits_{k \in \bar{\CK}}T_{k}$ seconds.

Now, as we aim to {minimize the delivery time}, the beamformer optimization problem can be formulated as $\underset{\{\Bv_{\CU}\}}\min \ \underset{k \in \bar{\CK}}\max \ T_{k}$, or equivalently as $\underset{\{\Bv_{\CU}\}}\max \ \underset{k \in \bar{\CK}}\min \ \frac{ R_{k}}{c_k}$. So, the weighted rate maximization for a given transmission can be formulated as
\begin{equation}
    \begin{aligned} 
	\label{opt:nonconvex_WMMF} \small	
		\underset{\substack{\{\Bv_{\CU}, \gamma^k_{\CU}, R_{\text{sum}}^{\CQ} \}}}{\max} & \underset{k \in \bar{\CK} }{\min} \ \underset{\substack{\CQ \subseteq {\CD_k}, \\ |\CQ| \neq 0}}{\min} \ 
        \frac{1}{c_k|\CQ|}R_{\text{sum}}^{\CQ}
	\\
	 \mathrm{s.  \ t.} \quad
       &R_{\text{sum}}^{\CQ} \leq \log\left(1 + \sum_{\CU \in \CQ}\gamma^{k}_{\CU}\right), \quad \forall k \in \bar{\CK}, \forall \CQ \subseteq {\CD_k}, |\CQ| \neq 0 \\
      &\gamma^{k}_{\CU} \leq \frac{|\Bh_{k}^{H}\Bv_{\CU}|^2}{\sum_{\CV \in \CI_k}|\Bh_{k}^{H}\Bv_{\CV}|^2 + \sigma^2}, \quad \forall k \in \bar{\CK}, \forall \CU \in {\CD}_k,
    \\
        &\sum_{\substack{\CU \subseteq \bar{\CK}, \\ |\CU| = \hat{t}+1}}\|\Bv_{\CU}\|^2\leq P_T,
\end{aligned}
\end{equation}
where $P_T$ is the total available power at the transmitter. Problem~\eqref{opt:nonconvex_WMMF} can be equivalently rewritten in epigraph form as
\begin{subequations}\label{opt:nonconvex_WMMF_epi}
\begin{align}
    \max_{\{\Bv_{\CU}, \gamma^k_{\CU}, R_{k}, R\}} &R \nonumber\\
    \textrm{s.t.} \quad &R \leq \frac{1}{c_k}R_k \, , \quad \forall k \in \bar{\CK} \; , \label{const: weighted rates}\\
    & R_k \leq \frac{1}{|\CQ|}\log\left(1 + \sum_{\CU \in \CQ}\gamma^{k}_{\CU}\right), \quad \forall k \in \bar{\CK}, \forall \CQ \subseteq {\CD_k}, |\CQ| \neq 0 \label{const:MAC_rate}\; ,   \\ 
    &\gamma^{k}_{\CU} \leq \frac{|\Bh_{k}^{H}\Bv_{\CU}|^2}{\sum_{\CV \in \CI_k}|\Bh_{k}^{H}\Bv_{\CV}|^2 + \sigma^2}, \quad \forall k \in \bar{\CK}, \forall \CU \in {\CD}_k,\label{const:MSE terms} \\
        &\sum_{\substack{\CU \subseteq \bar{\CK}, \\ |\CU| = \hat{t}+1}}\|\Bv_{\CU}\|^2\leq P_T \label{const:power}
\end{align}
\end{subequations}
{Note that $R$ is an auxiliary variable ensuring that user-specific rates $R_k$ are dedicated based on the corresponding weights $c_k$ (c.f.,~\eqref{const: weighted rates}). Moreover, condition~\eqref{const:MAC_rate} ensures that each user's dedicated rate $R_k$ lies within the MAC region. Auxiliary variables $\gamma^{k}_{\CU}$ are considered to help facilitate the convexification of conditions~\eqref{const:MAC_rate} and are limited to the message-specific SINR in~\eqref{const:MSE terms}, which is a non-convex constraint. Finally, constraint~\eqref{const:power} ensures that the dedicated power to the beamformers $\{\Bv_{\CU}\}$ does not exceed the available transmit power $P_T$.} 

Problem~\eqref{opt:nonconvex_WMMF_epi} is similar to the max-min optimization in~\cite{tolli2017multi} with the extra convex conditions~\eqref{const: weighted rates}. Thus, it can be efficiently solved following the same \ac{SCA} approach proposed in~\cite{tolli2017multi}. Here, we briefly recap the steps for the sake of completeness. First, Eq.~\eqref{const:MSE terms} is rewritten as
\begin{equation}\label{eq:MSE terms}
    \sum_{\CV \in \CI_k}|\Bh_{k}^{H}\Bv_{\CV}|^2 + \sigma^2 \leq \frac{\sum_{\CV \in \CI_k \cup \CU}|\Bh_{k}^{H}\bar{\Bv}_{\CV}|^2 + \sigma^2}{1+{\gamma}^{k}_{\CU}},
\end{equation}
and then, using the first order Taylor expansion, the right hand side of~\eqref{eq:MSE terms} is lower bounded by 
\begin{equation}\label{eq: MSE_approx}
\begin{aligned}
    \mathcal {L}(\mathbf {v}_{\mathcal {V}}, \mathbf {h}_{k}, \gamma _{\mathcal {U}}^{k}) \triangleq & \Big(\sum_{\CV \in \CI_k \cup \CU} \big( 2\mathbb{R}(\bar{\Bv}_{\CV}^H\Bh_{k}\Bh_{k}^{H}\Bv_{\CV}) -|\Bh_{k}^{H}\bar{\Bv}_{\CV}|^2\big)  \\ & \qquad- \frac{\sum_{\CV \in \CI_k \cup \CU}{|\Bh_{k}^{H}\bar{\Bv}_{\CV}|^2 + \sigma^2}}{1+\bar{\gamma}^{k}_{\CU}}({\gamma}^{k}_{\CU}-\bar{\gamma}^{k}_{\CU})+ \sigma^2 \Big) \frac{1}{{1+\bar{\gamma}^{k}_{\CU}}},
\end{aligned}
\end{equation}
where, $\bar{\Bv}_{\CV}$ and $\bar{\gamma}^{k}_{\CU}$ are the fixed approximation points. Now, substituting~\eqref{eq: MSE_approx} in~\eqref{eq:MSE terms}, we can approximate~\eqref{opt:nonconvex_WMMF_epi} as the following convex problem
\begin{subequations}\label{opt: convex_WMMF_epi}
\begin{align}
    \max_{\{\Bv_{\CU}, \gamma^k_{\CU}, R_{k}, R\}} &R \nonumber\\
    \textrm{s.t.} \quad &R \leq \frac{1}{c_k}R_k \, , \quad \forall k \in \bar{\CK} \; , \\
    & R_k \leq \frac{1}{|\CQ|}\log\left(1 + \sum_{\CU \in \CQ}\gamma^{k}_{\CU}\right), \quad \forall k \in \bar{\CK}, \forall \CQ \subseteq {\CD_k}, |\CQ| \neq 0 \; \label{eq:mac_region_condition},   \\ 
    &\sum_{\CV \in \CI_k}|\Bh_{k}^{H}\Bv_{\CV}|^2 + \sigma^2 \leq \mathcal {L}(\mathbf {v}_{\mathcal {V}}, \mathbf {h}_{k}, \gamma _{\mathcal {U}}^{k}), \quad \forall k \in \bar{\CK}, \forall \CU \in {\CD}_k,\\
        &\sum_{\substack{\CU \subseteq \bar{\CK}, \\ |\CU| = \hat{t}+1}}\|\Bv_{\CU}\|^2\leq P_T.
\end{align}
\end{subequations}
Finally, following the same approach as~\cite{tolli2017multi}, the beamformers $\Bv_{\CV}$ are found by iteratively solving~\eqref{opt: convex_WMMF_epi} until the convergence. 

\begin{remark}\label{remark: approximated delivery time}
Denoting the total delivery time of the proposed scheme with $T_T$, we have
\begin{equation}\label{eq:actual_delivery_time}
    T_T = \frac{1}{\binom{K-1}{\hat{t}+{\alpha}-1}\binom{\hat{t}+{\alpha}-1}{\hat{t}}}\underset{\substack{\bar{\CK} \subseteq [K], \\ |\bar{\CK}| = \hat{t}+{\alpha}}}\sum\max_{k \in \bar{\CK}}\frac{1-m(s_k)}{ R_{k}} \; .
\end{equation}
This simply follows the fact that every user $k$ needs to receive $1-m(s_k)$ units of data from the server during the delivery phase, and this data is delivered using $\binom{\hat{t}+{\alpha}-1}{\hat{t}}$ data terms in $\binom{K-1}{\hat{t}+{\alpha}-1}$ transmission vectors (cf. the proof of Lemma~\ref{theorm: delivery}).
\end{remark}

It should be noted that although the discussions so far applied to the delivery phase, we also need an approximation of the expected delivery time to optimize the memory allocation during the placement phase (c.f Section~\ref{sec:cache_placement}). However, during the placement phase, actual user locations are not known. Hence, the common global caching gain $\hat{t}$, the actual achievable rates $R_k$, and the actual delivery time $T_T$ can not be computed. To tackle this issue, we use an approximation of $T_T$ assuming uniform access probability for all the states as follows. 
\begin{lemma}\label{lemm:aprox_delivery_time}
The total delivery time $T_T$ calculated in~\eqref{eq:actual_delivery_time} can be approximated as
\begin{equation}\label{eq:T_T_approx}
    \hat{T}_T = \frac{K}{\bar{t}+{\alpha}}\max_{s \in \CS}\frac{1-m(s)}{r(s)} \; .
\end{equation}
\end{lemma}
\begin{proof}
Noting that the total dedicated rate to user $k$ for all the $D = \binom{\hat{t}+{\alpha}-1}{\hat{t}}$ messages is $DR_k$, we first substitute $DR_k$ with its upper bound $r(s_k)$ to approximate~\eqref{eq:actual_delivery_time} as
\begin{equation}\label{eq:T_T-first_approx}
    T_T \sim \frac{1}{\binom{K-1}{\hat{t}+{\alpha}-1}}\underset{\substack{\bar{\CK} \subseteq [K], \\ |\bar{\CK}| = \hat{t}+{\alpha}}}\sum\max_{k \in \bar{\CK}}\frac{1-m(s_k)}{r(s_k)} \; .
\end{equation}
Then, using inequality $\max\limits_{k \in \bar{\CK}}\frac{1-m(s_k)}{r(s_k)} \leq \max\limits_{s \in \CS}\frac{1-m(s)}{r(s)}$, we substitute the RHS of~\eqref{eq:T_T-first_approx} with its upper bound $\frac{\binom{K}{\hat{t}+{\alpha}}}{\binom{K-1}{\hat{t}+{\alpha}-1}}\max\limits_{s \in \CS}\frac{1-m(s)}{r(s)}$ to get  $\frac{K}{\hat{t}+{\alpha}}\max_{s \in \CS}\frac{1-m(s)}{r(s)}$. Finally, using inequality $\bar{t} \leq \hat{t}$, $T_T$ is approximated as~\eqref{eq:T_T_approx}.
\end{proof}

The delivery time approximation in~\eqref{eq:T_T_approx} can be further simplified by assuming that $R_w =  \frac{r(s)}{1-m(s)}$ is independent of the state $s$. The intuition behind this assumption is that with the proposed location-dependent cache placement and nested data delivery, the amount of data sent to each target user during every transmission is directly proportional to its delivery rate. With this assumption, we have $\hat{T}_T \sim \frac{K}{\hat{t}+{\alpha}} \frac{1}{R_w}$ and the symmetric rate will be $R_{w}^{s} = \frac{K}{T_T} = {(\hat{t}+{\alpha})R_{w}}$. Following a similar argument for the symmetric multi-antenna \ac{CC} scheme in~\cite{tolli2017multi}, the symmetric rate there could also be approximated as $R_u^s = \frac{K}{T_T} = {(t+{\alpha})R_u}$, where $R_{u} = \frac{\bar{r}}{1-M/S}$, $\bar{r}$ is the common max-min sum-rate, and $t = {KM}/{S}$ is the global caching gain (c.f.~\cite{tolli2017multi} Section~IV). This results in
\begin{equation}\label{eq:rate_ratio}
    \frac{R^{{s}}_{{w}}}{R^{{s}}_{{u}}} = \frac{(\hat{t}+{\alpha})R_{{w}}}{(t+{\alpha})R_{{u}}},
\end{equation}
which indicates that compared with the symmetric scheme of~\cite{tolli2017multi}, the proposed location-dependent scheme can improve the delivery time if the DoF loss resulting from the non-uniform cache placement (i.e., $\frac{\hat{t}+{\alpha}}{t+{\alpha}}\leq 1$) can be compensated by the rate improvement (i.e., $\frac{R_{{w}}}{R_{{u}}} \geq 1$) due to the multi-rate transmission support.

\subsection{Resolving the Imbalanced Global Caching Gain Bottleneck}
\label{Sec:phantom user}
Our proposed location-dependent \ac{CC} scheme enables the global caching gain of $\hat{t} = \min\limits_{k \in [K]} t_k$ to be achieved together with the spatial multiplexing gain of ${\alpha}$, while also addressing the wireless connectivity bottleneck problem. However, the $\min$ operation in $\hat{t}$ could cause the global caching gain to vanish if a subset of users were located at states with strong wireless connectivity (i.e., a subset of users have small $t_k$ values). This is an undesired effect as the users with better channel conditions limit the performance improvement enabled by the underlying multi-antenna \ac{CC} mechanism.

To address this issue, we use the \emph{phantom user} concept introduced in~\cite{salehi2021lowsubpacketization}. In a nutshell, phantom users are virtual, non-existent users that are assumed to be part of the network when the transmission codewords are designed, and their effect is removed before the actual transmission. Considering the fact that the global caching gain $\hat{t}$ is limited by users with strong wireless connectivity, the idea is to exclude such users from the \ac{CC}-aided (i.e., multicast) delivery phase and serve them through multi-user unicasting (i.e., considering spatial multiplexing and (possible) local caching gains only). Then, to make \ac{CC}-aided delivery work for the rest of the users, the excluded users are substituted by the same number of phantom users, all located in poor-connectivity states (hence, with large $t_k$ values). This results in an improved global caching gain for users with poor channel conditions, as the $\min$ operation would {no longer be} limited to users with strong wireless connectivity. As discussed in~\cite{salehi2020lowcomplexity}, the DoF loss caused by phantom users can be (partly) compensated through an improved beamforming gain enabled by optimized beamformers.

The enhanced multicast content delivery with phantom users is summarized in Algorithm~\ref{Alg:Delivery_phantom}. In this algorithm, we keep substituting users with best channel conditions with phantom users until the global caching gain $\hat{t}$ becomes larger than some threshold $t_{\text{target}}$,\footnote{In this paper, we have selected the $t_{\text{target}}$ value by experimenting. Finding the optimal value for this parameter is left for future research.} while also checking that the real achievable DoF for the remaining users does \emph{not} fall below $\hat{t} + {\alpha}$. The following example clarifies this procedure. 

\begin{exmp}
\label{exmp:interference-free-delivery_phant}
Consider a network in which a transmitter with {two} antennas (${\alpha = L}$) serves four users, where users are experiencing $t_1 = t_2 = t_3 =3$ and $t_4 = 1$ at the given time slot. Consequently, we have $\CT_1 = \{ A_{234}\}, \
    \CT_2 = \{ B_{134}\}, \
    \CT_3 = \{ C_{124}\}, \
    \CT_4 = \{ D_{1}, D_{2}, D_{3} \}.$
Following Algorithm~\ref{Alg:Delivery} for the considered users’ distribution, $\hat{t} = 1$ results in four transmissions where three users are served during each transmission. 
However, following Algorithm~\ref{Alg:Delivery_phantom}, user $4$ is first excluded from the set of \ac{CC}-aided users (i.e., $\CK_p = \{4\}$), and hence, the common global caching gain and the potential DoF increase to $\bar{t} = \min_{k \in \{1,2,3\}}t_k = 3$ and $\bar{t}+{\alpha} = 5$, respectively. However, since the actual remaining number of \ac{CC}-aided users is three, the real achievable DoF, in this case, remains equal to three, but the excess spatial multiplexing gain could be used for an enhanced beamforming gain (and better rate). Thus, using Algorithm~\ref{Alg:Delivery_phantom}, the data delivery is done in two consecutive transmissions, i.e.,  $\bar{\Bx}_{123} = \Bv_{123} \left(A_{234}*B_{134}*C_{124}\right),$ and $\bar{\Bx}_{4} = \Bv_{4} \left(\prod(D_{1}, D_{2}, D_{3})\right)$. 
Note that $\bar{\Bx}_{123}$ and $\bar{\Bx}_{4}$ will be delivered to user sets $\CU_1 = \{1,2,3\}$ and $\CU_{2} = \{4\}$ interference-free, respectively. 
\end{exmp}

\begin{algorithm}[t]
\small
    \caption{Phantom-based Multicast Content Delivery}
	\begin{algorithmic}[1]
		\Procedure{DELIVERY}{}
		
		\State $\hat{t} = \min_{k \in \CK} t_k$
		\If{$\hat{t} < t_{\text{target}}$}\label{alg:limiting_user_elimination_start}
		\State $\CK_p \gets \{ k \ | \ t(s_k) = \hat{t} \}$
		\If{$|\CK \setminus \CK_p| > \hat{t} + {\alpha}$} \label{alg:phantom_remaining_user_check}
		\State $\bar{t} = \min_{k \in \CK \setminus \CK_p} t_k$ \Else
		\State $\bar{t} \gets \hat{t}$ \& $\CK_p \gets \emptyset$
		\EndIf
		\EndIf\label{alg:limiting_user_elimination_end}
		\ForAll{$\Bar{\CK} \subseteq \CK : |\Bar{\CK}| = \bar{t} + {\alpha}$} 
		\If{$|\Bar{\CK} \setminus \CK_p| \geq {\alpha}$} \label{alg:phantom_DoF_check}
		\State $\Bx_{\Bar{\CK}} \gets 0$
		    \ForAll{$\CU \subseteq \Bar{\CK} : |\CU| = \bar{t}+1$}
		        \State $x_{\CU} \gets 0$
		        \ForAll{$k \in \CU \setminus \CK_p$}
		            \State ${\varphi_k} \gets \binom{t_k}{\bar{t}}\binom{K-\bar{t}-1}{{\alpha}-1}$, ${G_{\CU,k}} \gets 0$, $\CU_{-k} \gets \CU \backslash \{k \}$
		            \ForAll{$\CV_k \subseteq \CK : |\CV_k| = t_k+1$}
		                \If{$\CU_{-k} \subseteq \CV_k$, $k \not\in \CV_k$} 
		                \State $W_{\CV_k,k}^q \gets$ \textsc{Chunk}($W_{\CV_k,k}, {\varphi_k}$)
		                \State ${G_{\CU,k}} \gets$ \textsc{Concat} $ ({G_{\CU,k}}, W_{\CV_k,k}^q)$ 
		            \EndIf
		            \EndFor
		        \State $x_{\CU} \gets$ \textsc{Nest} $(x_{\CU}, {G_{\CU,k}}, R_k)$ 
		        \EndFor
		        \State $\Bx_{\Bar{\CK}} \gets \Bx_{\Bar{\CK}} + \Bv_{\CU}x_{\CU}$
		    \EndFor
		    \State Transmit $\Bx_{\Bar{\CK}}$
		    \EndIf
		\EndFor
		\EndProcedure
	\end{algorithmic}
	\label{Alg:Delivery_phantom}
\end{algorithm}

\section{Simulation Results}
\label{sec:Simulations}
The performance of the proposed location-dependent scheme is evaluated by numerical simulations. We consider an XR application in a bounded $30 \times 30 \mathrm{[m^2]}$ room, where a different 3D image is needed to rebuild the FoV in every tile of size $1 \times 1 \mathrm{[m^2]}$, resulting in $S = 900$ states. The requested data is served by a transmitter with ${L = 32}$ antennas {and spatial multiplexing gain of $\alpha\leq L$}, located in the middle of the room on the ceiling at the height of $5 \mathrm{[m]}$. The small-scale fading of the channel $\Bh_k$ 
is assumed to follow Rayleigh distribution, while the path loss for a user at state $s \in [S]$ is modeled as~\cite{lodro2018statistical}: 
\begin{equation*}
    PL(s) = 32.4[dB]+20\log_{10}(f) + 10\eta\log_{10}(d_s) + \zeta,
\end{equation*}
where $d_s$ is the distance between the center of the state $s$ and the transmitter, $\eta = 3$ is the path-loss exponent, and $f$ is the frequency. The term $\zeta \sim \mathbb{N}(0, {\sigma})$ with standard deviation $\sigma$ is used to model the impact of randomly-placed objects obstructing the propagation path between the transmitter and the receivers, similar to the shadowing effect in outdoor propagation environments. To compute the state-specific expected throughput $r(s)$ for the initial cache placement in~\eqref{cache-allocation_dif_conv},
the expectation in~ \eqref{eq: state-rate} is taken over all possible user locations and channel realizations in state $s$.
Unless otherwise mentioned, we assume {that} the transmit power is set such that the received \ac{SNR} at the room borders is equal to  $0 \mathrm{[dB]}$ (ignoring the `shadowing' effect $\zeta$). During the delivery phase, we assume every user $k \in [K]$ can be located at any state $s \in [S]$ with uniform probability. In all simulations, we use optimized beamformers obtained by solving~\eqref{opt:nonconvex_WMMF}.

As discussed throughout the paper, the proposed location-dependent CC scheme applies to application scenarios such as XR gaming, where the QoS of users is {affected by} the delivery time. Therefore, we do not just compare the average delivery time (rate) as done in most related CC studies (e.g., \cite{tolli2017multi}) but also consider the 95-percentile of the expected delivery time as a figure of merit. The performance of the following placement and delivery schemes are compared:
\begin{itemize}
    \item \textbf{Proposed, $\pmb{\phi \gg \frac{{\alpha}}{K}}$,w/unicasting}\footnote{{Note that `w/' and `w/o' are used to abbreviate `with' and `without,' respectively.}} where the placement phase is carried out by setting $\phi \gg \frac{{\alpha}}{K}$ in Alg.~\ref{Alg:placement}, but delivery is done by unicasting (i.e., by using the spatial multiplexing gain only and ignoring coded caching techniques);
    \item \textbf{Proposed, $\pmb{\phi \gg \frac{{\alpha}}{K}}$}, where $\phi \gg \frac{{\alpha}}{K}$ is used in Alg.~\ref{Alg:placement} and delivery is performed using Alg.~\ref{Alg:Delivery_phantom};
    \item \textbf{Proposed, $\pmb{\phi = \frac{{\alpha}}{K}}$}, where $\phi = \frac{{\alpha}}{K}$  is used in Alg.~\ref{Alg:placement} (i.e., considering multicast content delivery already in the placement phase) and delivery is done using Alg.~\ref{Alg:Delivery_phantom};
    \item \textbf{MS}, where the cache placement is uniform (i.e., $\phi \ll \frac{{\alpha}}{K}$  is used in Alg.~\ref{Alg:placement}), and the baseline content delivery algorithm of~\cite{tolli2017multi} is used.
\end{itemize}

We first compare different schemes based on their delivery times accumulated over $500$ random user drops. Fig.~\ref{fig:sim_cdf} plots the cumulative distribution function (CDF) of the total delivery times for all realizations.
For reference, we have also added simulation results for two other schemes, which are very similar to \textit{Proposed, ${\phi \gg \frac{{\alpha}}{K}}$} and \textit{Proposed, ${\phi = \frac{{\alpha}}{K}}$} but use Algorithm~\ref{Alg:Delivery} for data delivery (i.e., they don't use phantom users to address the caching gain bottleneck as discussed in Sec.~\ref{Sec:phantom user}). As can be seen, incorporating phantom users always improves the performance; hence, throughout the rest of the text, we always assume {that} Algorithm~\ref{Alg:Delivery_phantom} is used for data delivery.
From Fig.~\ref{fig:sim_cdf} it is clear that the \textit{MS} scheme
has the largest variation in total delivery time among all the schemes, which is undesirable in our considered use cases with location-dependent content requests (e.g., XR gaming). The reason for this considerable variation is that the \textit{MS} scheme only intends to maximize the global caching gain, which results in better performance than other schemes when all the users have good channel conditions but deteriorates the rate when a subset of users experience poor connectivity. On the other hand, our proposed schemes provide a robust performance by keeping the variance in delivery time very small, with the \textit{Proposed, ${\phi = \frac{{\alpha}}{K}}$} scheme providing the best results. This robust performance is a direct result of the proposed non-uniform cache placement, as it makes the algorithm immune to wireless connectivity bottleneck areas by allocating more memory to store the content of such areas.

\begin{figure}
    \centering 
    \includegraphics[width=0.6\columnwidth,keepaspectratio]{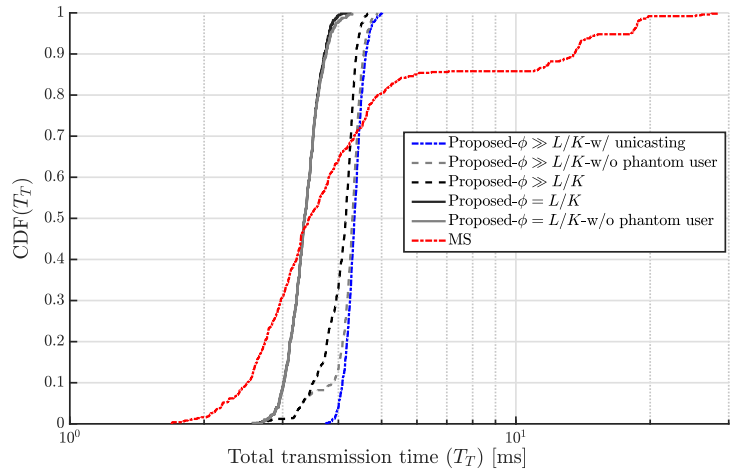}
    \caption{The cumulative distribution function (CDF) of total delivery time (logarithmic-scale) for $500$ realizations, where $K = 6, M/S = 0.33, {\alpha}=2$ and ${\sigma = 8}$.}
    \label{fig:sim_cdf}
\end{figure}

Figures~\ref{fig:sim_sigma_time} and \ref{fig:sim_sigma_95time} compare the performance of different schemes 
with respect to the standard deviation of obstructed locations parameter ($\sigma$). As illustrated, for small $\sigma$ (i.e., less variation in large-scale fading among states), the traditional \textit{MS} scheme outperforms other methods. This is because in our proposed schemes, we sacrifice the global caching gain (i.e., $t = \frac{KM}{S}$) for a higher local caching gain (i.e., $m_i$), which results in a 
better transmission rate for individual users (the $\frac{R_w}{R_u}$ ratio in~\eqref{eq:rate_ratio}) but reduces the number of users served simultaneously (the $\frac{\hat{t}+{\alpha}}{t+{\alpha}}$ ratio in~\eqref{eq:rate_ratio}). The rate improvement for individual users is more prominent when they experience relatively poor channel quality, which is not the case when $\sigma$ is small. However, as $\sigma$ becomes larger (i.e., there are more attenuated states), 
the \textit{MS} scheme performs worse than the proposed schemes. This is because, with larger $\sigma$, users are more likely to experience poor connectivity, increasing the effectiveness of the local caching gain in decreasing the total delivery time. 
It should be noted that in both {Figures}~\ref{fig:sim_sigma_95time} and~\ref{fig:sim_sigma_time}, if $\sigma > {10}$, the \textit{Proposed, $\phi \gg \frac{{\alpha}}{K}$} outperforms all other schemes. This is because, in that regime, the large variety in the expected achievable rate of different states forces the memory allocation to become more non-uniform. As a result, the minimum achievable global caching gain (i.e., $\hat{t}$) becomes very small, and we need to rely more on phantom users to resolve the imbalance global caching gain bottleneck (c.f. Section~\ref{Sec:phantom user}). However, phantom users work better with the \textit{Proposed, $\phi \gg \frac{{\alpha}}{K}$} scheme as it allows even more non-uniform memory allocations than the \textit{Proposed, $\phi = \frac{{\alpha}}{K}$} scheme.


\begin{figure*}[tb]
\begin{minipage}[c]{0.49\textwidth}
\centering 
\setlength\abovecaptionskip{-0.25\baselineskip}
\includegraphics[width=1\columnwidth,keepaspectratio]{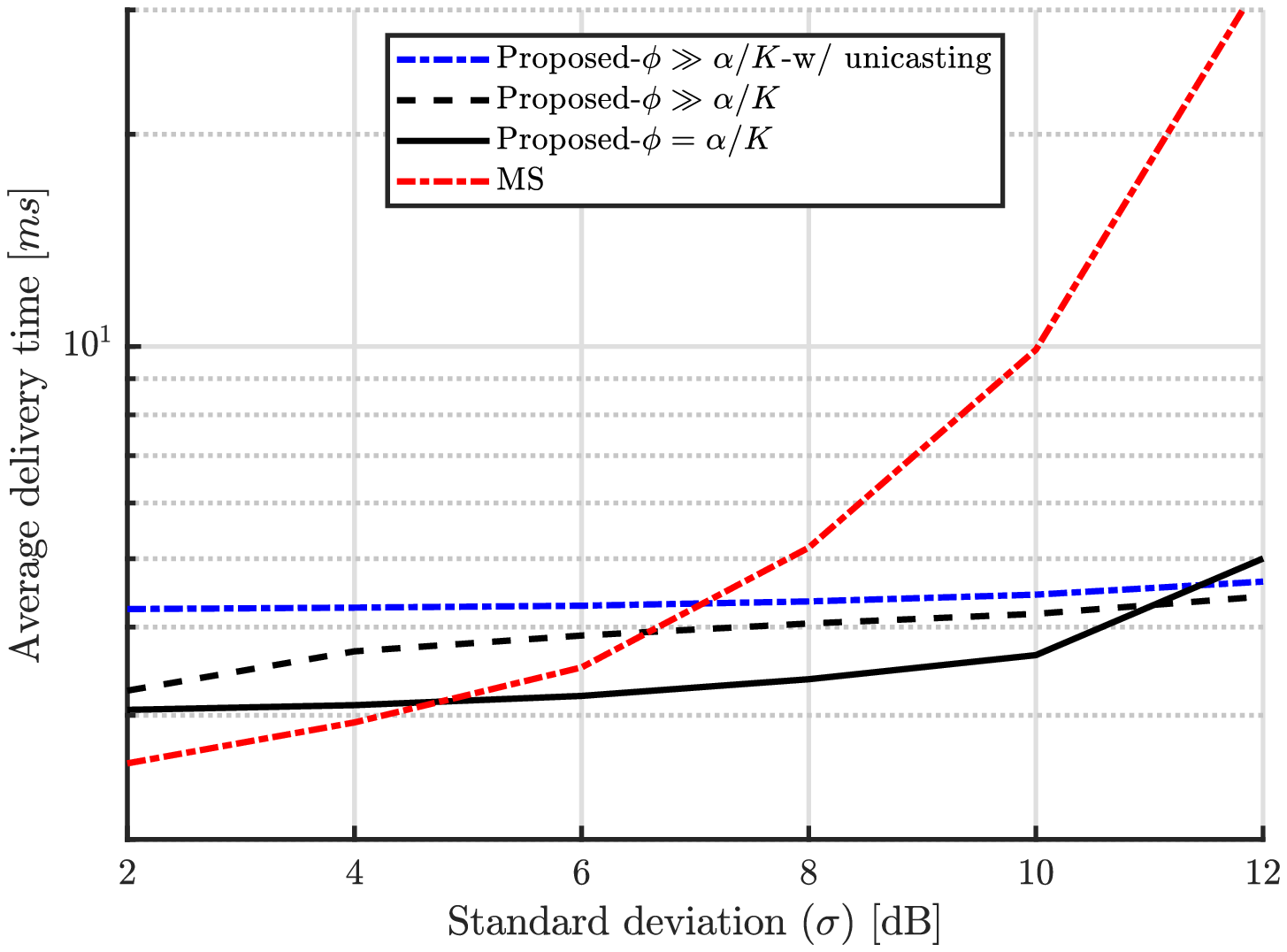}
\caption{Delivery time (logarithmic-scale) versus the standard deviation ($\sigma$), where $K = 6, M/S = 0.33$, and ${\alpha}=2$.}
\label{fig:sim_sigma_time}
\end{minipage}
\hspace{2mm}
\begin{minipage}[c]{0.49\textwidth}
\centering 
\setlength\abovecaptionskip{-0.25\baselineskip}
\includegraphics[width=1\columnwidth,keepaspectratio]{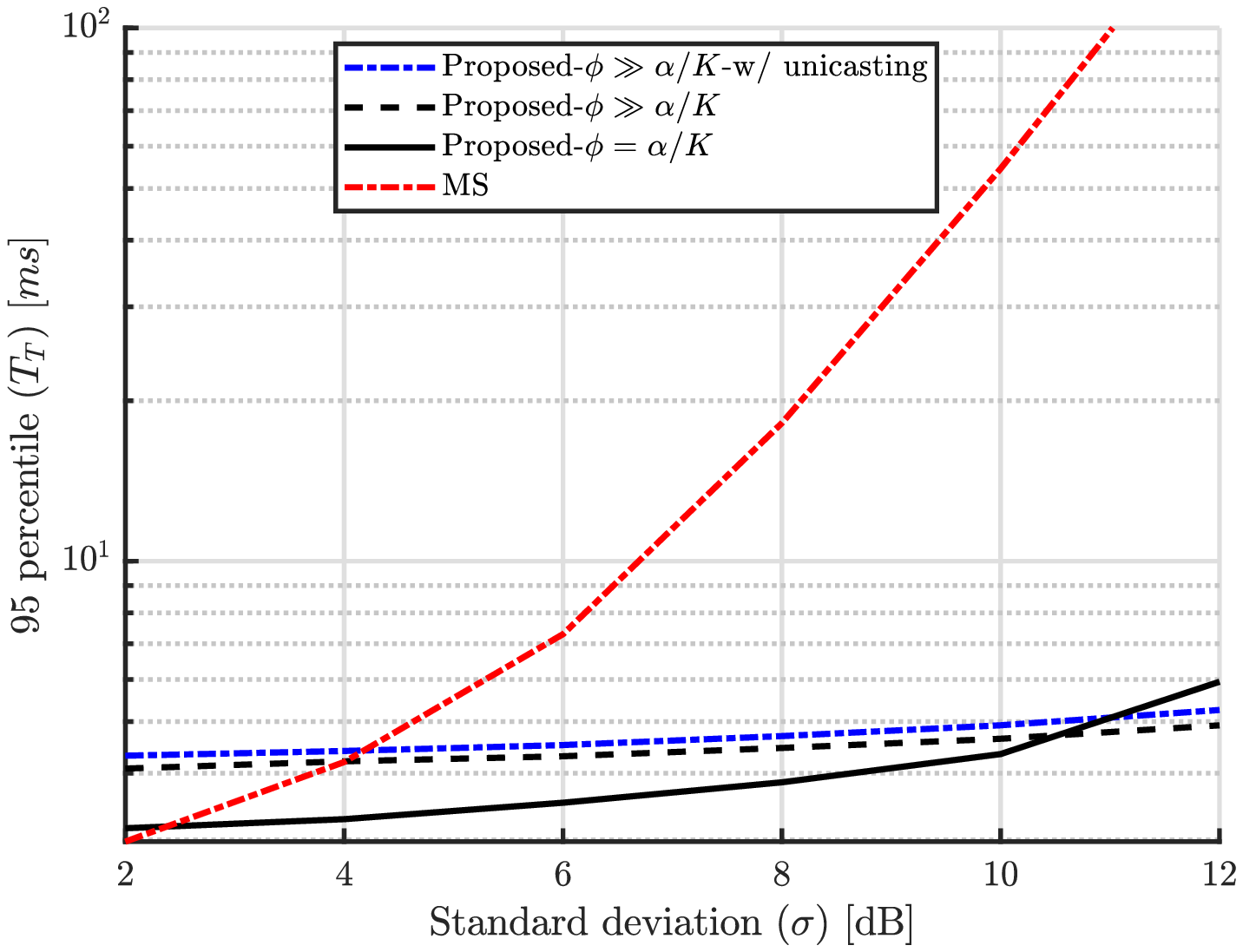}
\caption{Delivery time (logarithmic-scale) versus $\sigma$, where $K = 6, M/S = 0.33$, and ${\alpha}=2$.}
\label{fig:sim_sigma_95time}
\end{minipage}
\end{figure*}

Figures~\ref{fig:sim_power_95time} and~\ref{fig:sim_L_95time} compare the performance of different methods with respect to the SNR value at the room border and the {spatial multiplexing gain} ${\alpha}$, respectively. As shown in Fig.~\ref{fig:sim_power_95time}, for smaller SNR values at the room border, the performance gap between the proposed schemes and the \textit{MS} scheme widens. This is because with smaller SNR values (i.e., smaller transmit power), the achievable rate in different states gets highly affected by large- and small-scale fading, resulting in larger variations in the achievable rate of different users.
As a result, the proposed schemes that use non-uniform cache placement to compensate for such variations perform better than the \textit{MS} scheme. On the other hand, {for higher} received SNR at the cell edge, {performance gap between \textit{MS} and other schemes decreases} as the uniform memory allocation becomes {almost} optimal. Also, as illustrated in~Fig.~\ref{fig:sim_L_95time}, for a larger {spatial multiplexing gain},
our proposed schemes perform better than \textit{MS}. The reason is that, with a larger ${\alpha}$, the DoF value gets less sensitive to the global caching gain, i.e., the $\frac{\hat{t}+{\alpha}}{t+{\alpha}}$ ratio in~\eqref{eq:rate_ratio} converges to one. As a result, the rate improvements for individual users due to the increased local caching gain of the proposed schemes become more effective in reducing the total delivery time.\footnote{{Note that in Fig. 8, as the number of users is fairly small, it quickly limits the DoF $\min\{\alpha+\hat{t}, K\}$ as $\alpha$ is increased. For example, for the \textit{MS} scheme, the DoF is capped at six when $\alpha \geq 4$. 
}} 
\begin{figure*}[tb]
\begin{minipage}[c]{0.5\textwidth} 
\centering 
\setlength\abovecaptionskip{-0.25\baselineskip}
\includegraphics[width=1\columnwidth,keepaspectratio]{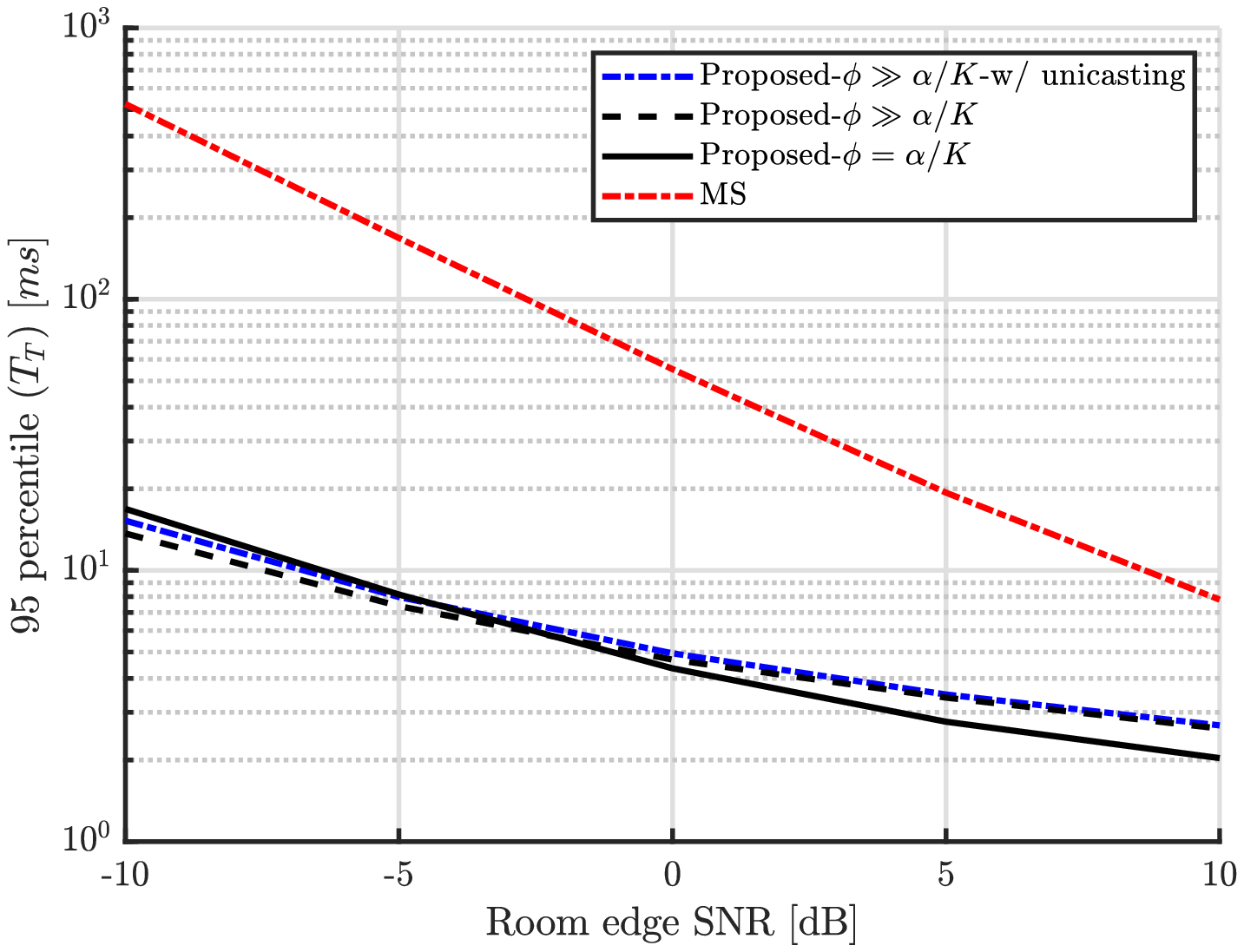}
\caption{95 percentile delivery time {(logarithmic-scale)} versus different room edge SNR,  {for} $K = 6, M/S = 0.33, {\sigma=10}$, and {$\alpha=2$}.}
\label{fig:sim_power_95time}
\end{minipage}
\hspace{1mm}
\begin{minipage}[c]{0.5\textwidth} 
\centering 
\setlength\abovecaptionskip{-0.25\baselineskip}\includegraphics[width=1\columnwidth,keepaspectratio]{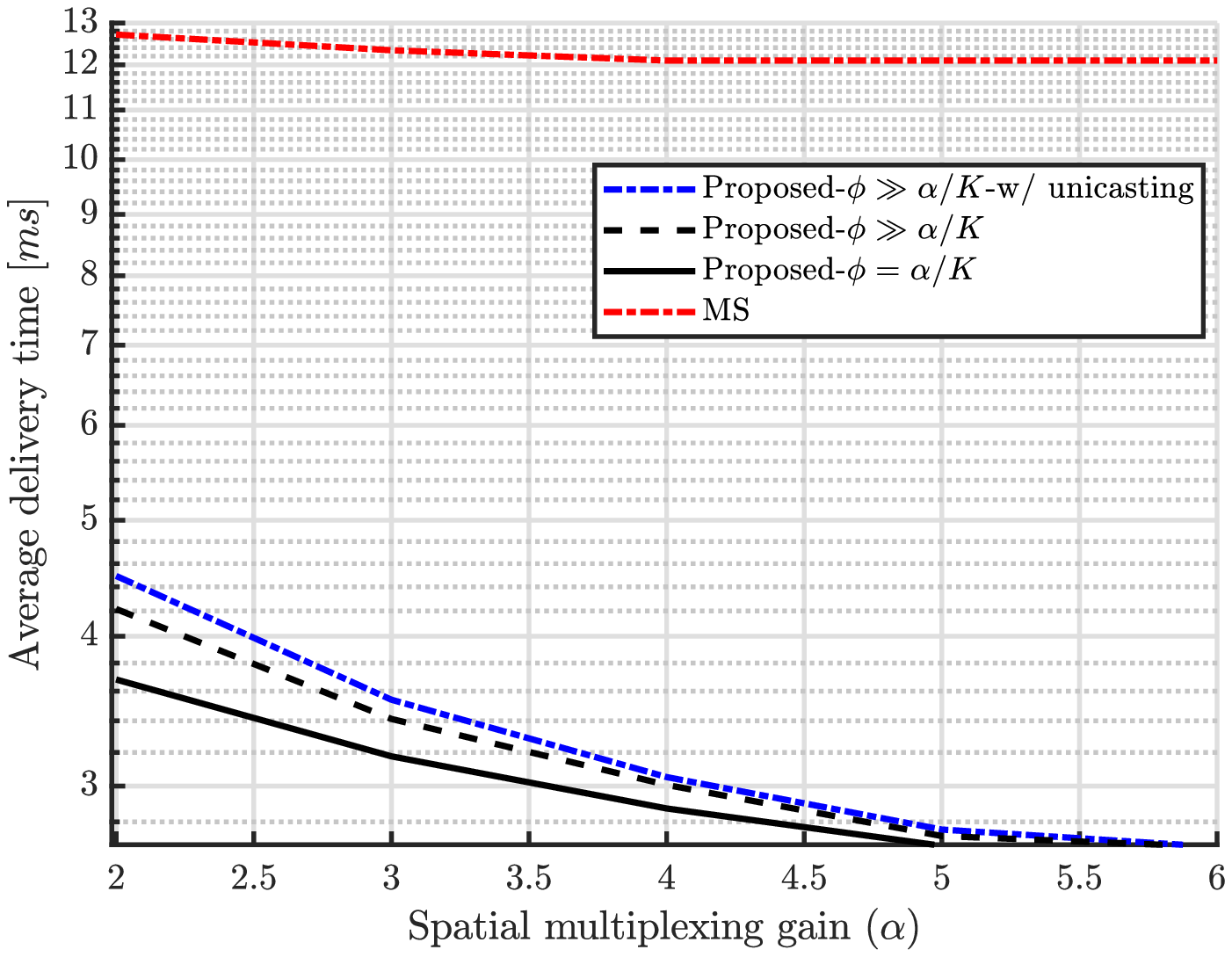}
\caption{Delivery time {(logarithmic-scale)} versus the {spatial multiplexing gain ($\alpha$)}, where $K = 6, M/S = 0.33$, and ${\sigma=10}$.}
\label{fig:sim_L_95time}
\end{minipage}
\end{figure*}  

In Fig.~\ref{fig:sim_M_95time}, we have compared the performance of different schemes with respect to the (normalized) available cache memory at the users $\frac{M}{S}$. As depicted, the performance gap between the proposed schemes and the \emph{MS} scheme 
{grows rapidly at first but then narrows as $\frac{M}{S}$ is increased. The reason for this behavior is that when $\frac{M}{S}$ is small, performance improvement is due to the local caching gain; i.e., all the available memory is used to cover wireless connectivity bottlenecks. However, as more memory becomes available, we reach a point where all the bottleneck areas are covered, and cached data starts to be used to improve the global caching gain. In fact, as the amount of available memory grows,}
%
the result of the memory allocation process in~\eqref{cache-allocation} gets closer to the uniform allocation of the \textit{MS} scheme.
%
Finally, Fig.~\ref{fig:sim_K_95time} shows how the user count parameter $K$ affects the delivery time of various schemes. As depicted, since a larger $K$ also means more data to be delivered, the delivery time generally grows with the number of users. However, since the global caching gains also scale with $K$, the number of users served in parallel (i.e., $\hat{t}+{\alpha}$) also increases for larger $K$, resulting in an overall performance improvement for all the {CC-based} schemes. Still, the \textit{Proposed, $\phi = \frac{{\alpha}}{K}$} scheme provides the best performance among all the schemes, and its required delivery time is affected minimally by the increase in $K$. This is because, with larger $K$, there is a higher chance of having users with poor connectivity, increasing the effectiveness of the improved local caching gain resulting from the underlying non-uniform cache placement.
%

\begin{figure*}[tb]
\begin{minipage}[c]{0.5\textwidth}
\centering 
\setlength\abovecaptionskip{-0.25\baselineskip}
\includegraphics[width=1\columnwidth,keepaspectratio]{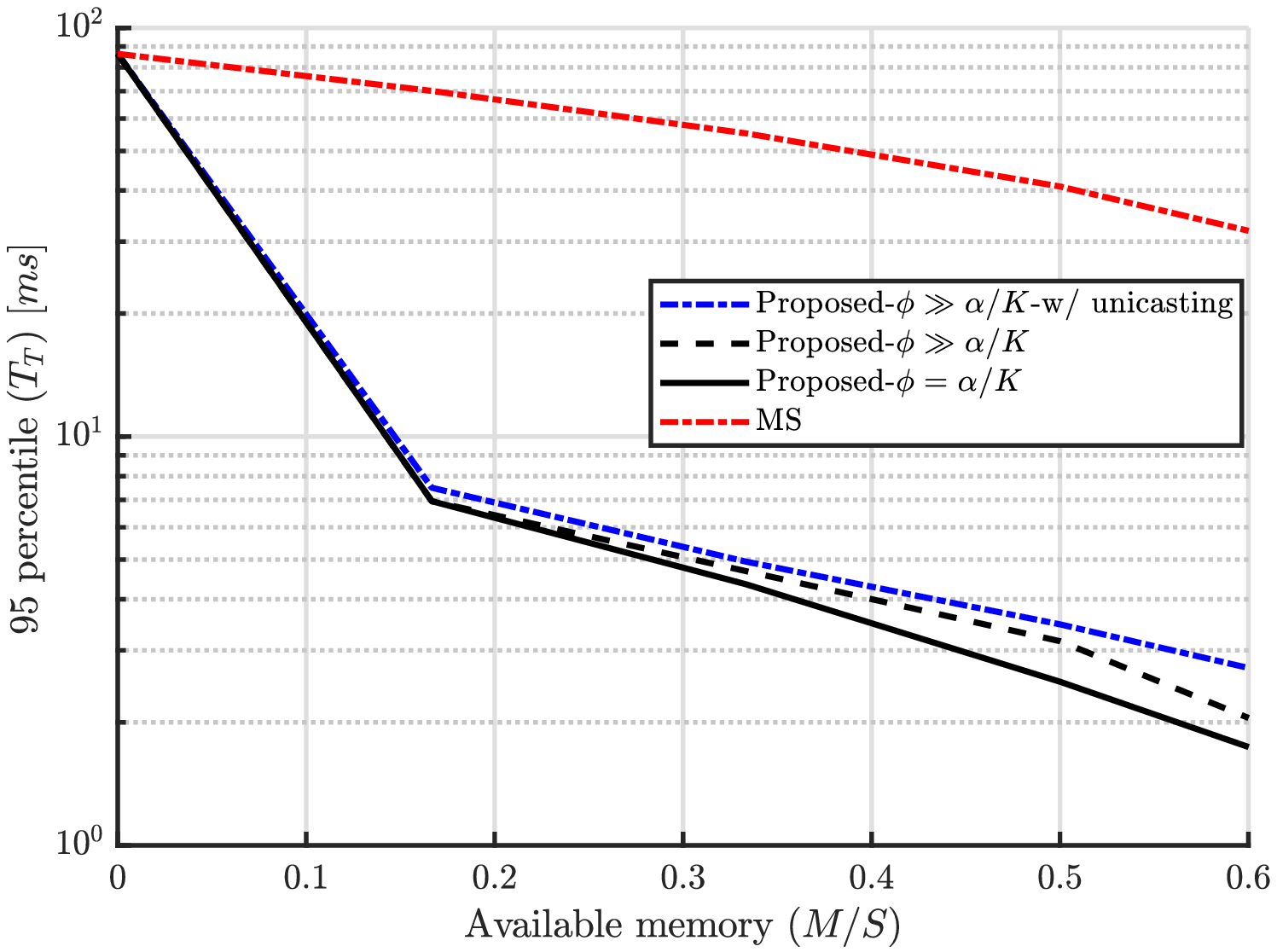}
\caption{Delivery time {(logarithmic-scale)} versus memory size ($M/S$), where $K = 6, {\sigma=10}$, and ${\alpha=2}$.}
\label{fig:sim_M_95time}
\end{minipage}
\hspace{2mm}
\begin{minipage}[c]{0.5\textwidth}
\centering 
\setlength\abovecaptionskip{-0.25\baselineskip}
\includegraphics[width=1\columnwidth,keepaspectratio]{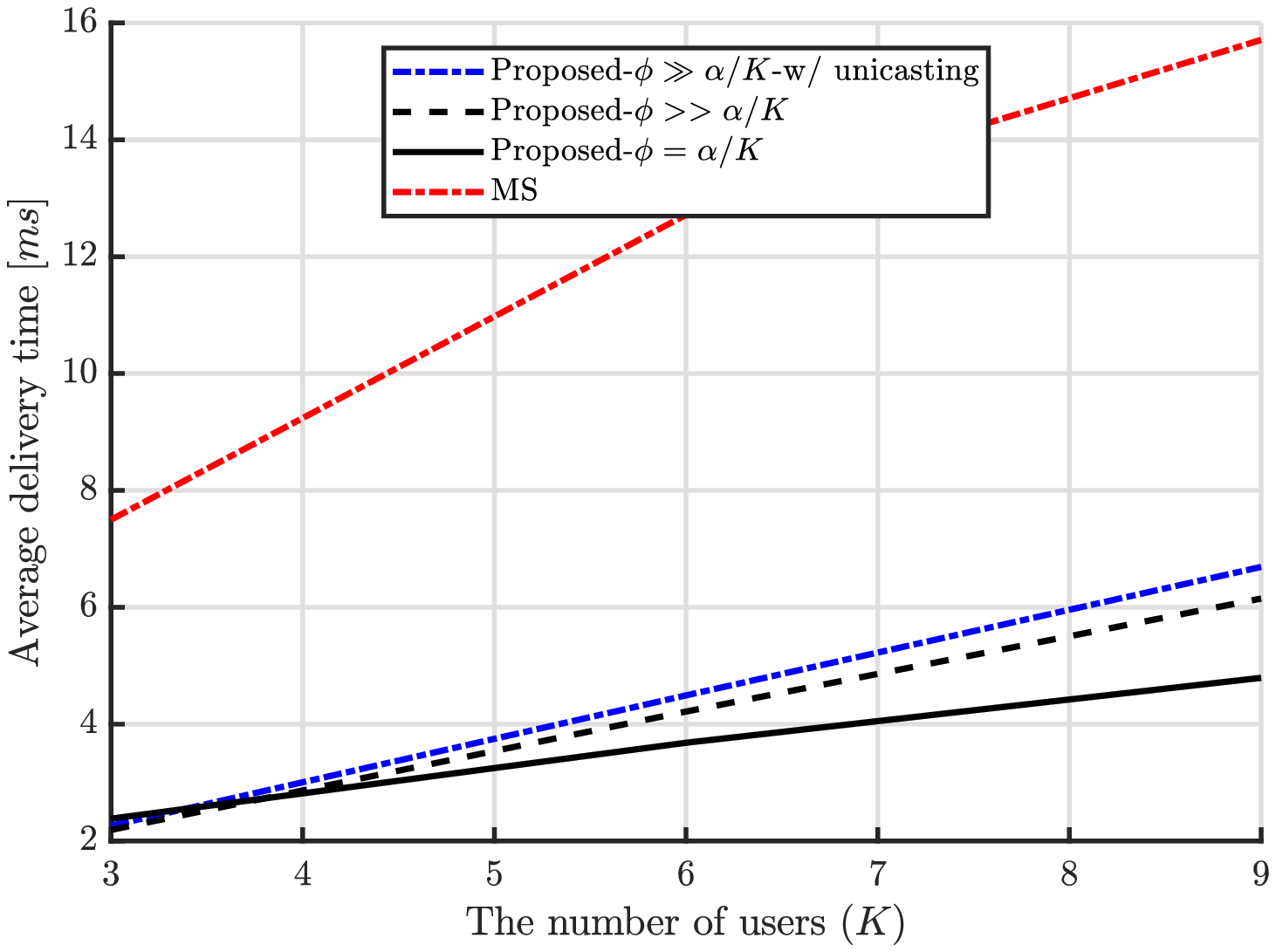}
\caption{Average delivery time versus user count ($K$), where $M/S = 0.33,  {\sigma=10}$, and ${\alpha=2}$.}
\label{fig:sim_K_95time}
\end{minipage}
\end{figure*} 

\section{Conclusion and Future Work}
\label{sec:conclusions}
 
A centralized location-dependent coded caching scheme with multi-rate content delivery, tailored for future XR applications, was proposed in this paper. Initially, the area of interest was divided into many small states such that the achievable rate at every point in each state could be considered the same. Then, based on each state's approximated achievable delivery rate, a memory loading process was performed to reduce the burden on wireless resources to serve ill-conditioned locations. This resulted in an uneven memory allocation, where larger cache portions were allocated to the contents requested in poor-connectivity states. Then, during the content delivery phase, a novel algorithm based on coded caching and multi-rate transmission techniques were devised to enable combined global caching and spatial multiplexing gains at the transmitter. Finally, the proposed method was shown to outperform the state-of-the-art in ill-conditioned scenarios where the ratio between the best and worst channel conditions was large. Future research opportunities include supporting multiple transmitters, {imperfect channel knowledge during the placement and delivery phases,} incorporating side information on user movement patterns and state transition probabilities, and considering a more dynamic scenario where the users' cache content is updated as they move throughout the application environment.

\appendices
\section{Non-integer global coded caching gain} \label{sec: Appendix A}
For the non-integer \emph{global coded caching gain} case (i.e., the case where $t(s) = Km(s)$ is non-integer), the proposed memory-sharing scheme in~\cite{MaddahAli-2014} is adopted for the cache arrangement:
\begin{itemize}
    \item[1.] The content file for state $s$, i.e., $W(s)$, is first divided into two non-overlapping parts $W_1(s)$ and $W_2(s)$, where $|W_1(s)| = \left(\lfloor t(s) \rfloor +1 - t(s)\right)|W(s)|$ and $|W_2(s)| = \left(t(s)-\lfloor t(s) \rfloor\right)|W(s)|$.
    \item[2.] Each user caches the data of $W_1(S)$ with integer $\underline{t}(s) = \lfloor t(s) \rfloor$, and data of $W_2(S)$ with integer $\bar{t}(s) = \lfloor t(s) \rfloor +1$, following the cache placement scheme in {Section}~\ref{sec:cache_placement}.
\end{itemize}
It is easy to verify that the proposed memory-sharing process does not violate the cache constraint, i.e., $\frac{\binom{K-1}{\underline{t}(s)-1}}{\binom{K}{\underline{t}(s)}}(\underline{t}(s)+1-t(s)) + \frac{\binom{K-1}{\underline{t}(s)}}{\binom{K}{\underline{t}(s)+1}}(t(s)-\underline{t}(s)) = \frac{t(s)}{K} = m(s)$. The common global coded caching gain is first computed as $\hat{t} = \min_{k \in \CK} \lfloor t_k \rfloor$ for the delivery phase. Then, the same file division procedure is followed for each file part $W_1(s_k)$ and $W_2(s_k)$ separately. To this end, every \textit{sub-file} of $W_1(s_k)$ and $W_2(s_k)$ is divided into $\alpha^1_k = \binom{K-\hat{t}-1}{L-1}\binom{\underline{t}_k}{\hat{t}}$ and $\alpha^2_k = \binom{K-\hat{t}-1}{L-1}\binom{\bar{t}_k}{\hat{t}}$ smaller segments, respectively. Then,  $\chi^1_k = \binom{K-\hat{t}-1}{\underline{t}_k - \hat{t}}$ file segments of $W_1(s)$ and $\chi^2_k = \binom{K-\hat{t}-1}{\bar{t}_k - \hat{t}}$ file segments of $W_2(s)$ are concatenated together to create a packet for user $k$. The following example clarifies the memory-sharing process.
\begin{exmp}
\label{exmp: non-int-t}
Consider a network similar to the one in example~\ref{exmp:placement}, where the only difference is $t(1) = 1.2$. As a result, $W(1)$ is first divided into $W_1(1)$ and $W_2(1)$, with size $|W_1(1)| = 0.8*400 = 320$ Megabytes and $|W_2(1)| = 0.2*400 = 80$ Megabytes. Then, $W_1(1)$ and $W_2(1)$ are cached based on the corresponding tables for $s = 1$ and $s = 2$ in Figure~\ref{fig:cache pool}. Now, consider the same location realization for the users as in example~\ref{exmp:interference-free-delivery}. The missing file parts at user $1$ in this case are $\CT_1 = \{A_{1,2}, A_{1,3}, A_{1,4}, A_{2,23}, A_{2,24}, A_{2,34}\}$. Then, based on the procedure mentioned above, the {the first} transmission vector is now built as
 \begin{equation}\small \nonumber
    \begin{aligned} 
        \Bx_{123} =& \Bv_{12} \left(\prod(A_{1,2}^{1}, A_{2,23}^{1}, A_{2,24}^{1}) *\prod(B_{13}^{1} , B_{14}^{1})\right) + \Bv_{13} \left(\prod(A_{1,3}^{1}, A_{2,23}^{2}, A_{2,34}^{1}) *\prod(C_{12}^{1} , C_{14}^{1})\right) \\ & \quad \quad + \Bv_{23} \left(\prod(B_{13}^{2} , B_{34}^{1}) *\prod(C_{12}^{2} , C_{24}^{1})\right). 
    \end{aligned}
\end{equation}
{The remaining three messages $\Bx_{124}, \Bx_{134}$, and $\Bx_{234}$ are also built similarly.}
 \end{exmp}

Finally, following a similar argument as in Theorem~\ref{theorm: delivery}, we can show that for a non-integer $t(s_k)$, user $k$ receives its missing files entirely. Specifically, each transmitted packet to user $k$ contains $\chi^1_k$ file segments from $W_1(s_K)$, each with size $ {1}/{\binom{K}{\underline{t}_k}\alpha^1_k}$ data units, and $\chi^2_k$ file segments from $W_2(s_K)$, each with size $ {1}/{\binom{K}{\bar{t}_k}\alpha^2_k}$ data units. As a result, considering the total number of received packets (i.e., $\binom{K-1}{\hat{t}+L-1}\binom{\hat{t}+L-1}{\hat{t}})$, user $k$ is served by the following data size

\begin{equation}\nonumber
    \frac{ \binom{K-1}{\hat{t}+L-1}\binom{\hat{t}+L-1}{\hat{t}} \binom{K-\hat{t}-1}{\underline{t}_k - \hat{t}} }{ \binom{K}{\underline{t}_k} \binom{\underline{t}_k}{\hat{t}}  \binom{K-\hat{t}-1}{L - 1}}(\underline{t}_k+1-t_k) + \frac{ \binom{K-1}{\hat{t}+L-1}\binom{\hat{t}+L-1}{\hat{t}} \binom{K-\hat{t}-1}{\bar{t}_k - \hat{t}} }{ \binom{K}{\bar{t}_k} \binom{\bar{t}_k}{\hat{t}}  \binom{K-\hat{t}-1}{L - 1}}(t_k-\underline{t}_k) = \frac{K-t_k}{K} = 1 - m(s_k) \; .
\end{equation}



\end{document}